\documentclass[pra,aps,twocolumn,groupedaddress,superscriptaddress,nofootinbib,preprintnumbers]{revtex4-2}

\usepackage[usenames,dvipsnames,svgnames,table]{xcolor}

\usepackage{graphicx}
\usepackage[nointegrals]{wasysym} %
\usepackage[export]{adjustbox}
\usepackage{amsmath,amsfonts,amssymb,latexsym}
\usepackage{hhline}
\usepackage{bm}
\usepackage{verbatim}
\usepackage{enumitem}
\usepackage{mathrsfs}
\usepackage{slashed}
\usepackage{empheq}

\newcommand{\Tr}{\text{Tr}}

\newcommand{\p}{\partial}

\newcommand{\lan}{\langle}
\newcommand{\ran}{\rangle}

\newcommand{\cmark}{\ding{51}}%
\newcommand{\xmark}{\ding{55}}%

\newcommand{\ra}{\rightarrow}

\newcommand{\wt}{\widetilde}

\newcommand{\uva}{{\mathbf{\hat a}}}
\newcommand{\uvb}{{\mathbf{\hat b}}}

\newcommand{\uvx}{{\mathbf{\hat x}}}

\newcommand{\uvz}{{\mathbf{\hat z}}}

\newcommand{\bfzero}{{\mathbf{0}}}

\renewcommand{\(}{\left(}
\renewcommand{\)}{\right)}

\newcommand{\mt}{\mapsto}

\newcommand{\tp}{\otimes}

\newcommand{\D}{\nabla}

\newcommand\bpm            {\begin{pmatrix}}
	\newcommand\epm           {\end{pmatrix}}

\newcommand{\ms}{\medskip}

\def\app#1#2{%
	\mathrel{%
		\setbox0=\hbox{$#1\sim$}%
		\setbox2=\hbox{%
			\rlap{\hbox{$#1\propto$}}%
			\lower1.1\ht0\box0%
		}%
		\raise0.25\ht2\box2%
	}%
}

\newcommand{\tw}{\textwidth}

\newcommand{\ct}{\Theta}

\newcommand{\inv}{^{-1}}

\newcommand{\ope}\odot

\usepackage{manfnt}

\newcommand{\bi}{\begin{itemize}}
	\newcommand{\ei}{\end{itemize}}

\usepackage{amsthm}
\newtheorem{theorem}{Theorem}
\newtheorem{definition}{Definition}
\newtheorem{corollary}{Corollary}
\newtheorem{proposition}{Proposition}

\newtheorem{lemma}{Lemma}

\newtheorem{remark}{Remark}
\theoremstyle{definition}

\newcommand\bpro		  {\begin{proposition}}
	\newcommand\epro 		  {\end{proposition}}
\newcommand\bproof			  {\begin{proof}}
	\newcommand\eproof 		  {\end{proof}}

\newcommand\ed            {\end{definition}}

\newcommand\be            {\begin{equation}}
\newcommand\ee            {\end{equation}}
\newcommand\ba            {\begin{aligned}}
\newcommand\ea            {\end{aligned}}
\newcommand\bea{\begin{equation}\begin{aligned}}
	\newcommand\eea{\end{aligned}\end{equation}}
\usepackage{xcolor}

\usepackage{hyperref} 
\hypersetup{final}
\definecolor{darkblue} {rgb}{.1,.5,0.65}
\definecolor{darkgreen}{rgb}{.1,.18,.82}
\hypersetup{colorlinks,
urlcolor   = blue,         
linkcolor  = darkgreen,     
citecolor  = darkblue,    
}
\newcommand{\sss}{\subsubsection}
\renewcommand{\ss}{\subsection}

\renewcommand{\a}{\alpha}
\renewcommand{\b}{\beta}
\renewcommand{\d}{\delta}
\newcommand{\De}{\Delta}
\newcommand{\g}{\gamma}
\newcommand{\G}{\Gamma}
\newcommand{\s}{\sigma}

\newcommand{\ep}{\varepsilon} %
\renewcommand{\l}{\lambda}

\renewcommand{\o}{\omega}
\renewcommand{\O}{\Omega}

\renewcommand{\r}{\rho}

\newcommand{\z}{\zeta}

\newcommand{\bfxi}{{\boldsymbol{\xi}}}

\newcommand{\bfF}{\mathbf{F}}

\newcommand{\bfd}{\mathbf{d}}

\newcommand{\bfr}{\mathbf{r}}
\newcommand{\bfs}{\mathbf{s}}

\newcommand{\bfu}{\mathbf{u}}
\newcommand{\bfv}{\mathbf{v}}

\newcommand{\bfx}{\mathbf{x}}

\newcommand{\zt}{\mathbb{Z}_2}

\newcommand{\EE}{\mathbb{E}}

\newcommand{\rr}{\mathbb{R}}
\newcommand{\qq}{\qquad}

\newcommand{\zz}{\mathbb{Z}}

\newcommand{\mcc}{\mathcal{C}}

\newcommand{\mce}{\mathcal{E}}

\newcommand{\mcd}{\mathcal{D}}
\newcommand{\mcl}{\mathcal{L}}

\newcommand{\mcs}{\mathcal{S}}

\newcommand{\mcr}{\mathcal{R}}

\newcommand{\sfA}{\mathsf{A}}

\newcommand{\sfC}{\mathsf{C}}

\newcommand{\sfN}{\mathsf{N}}

\newcommand{\sfP}{\mathsf{P}}

\newcommand{\sfT}{\mathsf{T}}

\newcommand{\sfa}{\mathsf{a}}
\newcommand{\sfb}{\mathsf{b}}
\newcommand{\sfc}{\mathsf{c}}
\newcommand{\sfd}{\mathsf{d}}

\usepackage[mathscr]{eucal} %

\newcommand{\scc}{\mathscr{C}}

\newcommand{\scg}{\mathscr{G}}

\usepackage{dcolumn}

\usepackage{pifont}
\usepackage[normalem]{ulem}

\usepackage[caption=false]{subfig}
\usepackage{enumitem}
\usepackage{amsthm}
\usepackage{physics}
\usepackage{booktabs}
\usepackage{algorithm}
\usepackage{algpseudocode}
\usepackage{bm}
 
\renewcommand\qq{\qquad}

\newcommand{\oEE}{\mathop{\mathbb{E}}}
\newcommand{\ethan}[1]{ { \color{blue} \footnotesize \textsf{ethan: \textsl{#1}} }}

\newcommand{\trel}{t_{\rm rel}}
\newcommand{\tmem}{t_{\rm mem}}

\usepackage{seqsplit}
\usepackage{array}
\newcommand{\plog}{{p_{\sf log}}}

\renewcommand{\trel}{t_{\sf mem}}
\renewcommand{\tmem}{t_{\sf mem}}

\newcommand{\tdec}{T_{\sf dec}}
\newcommand{\emp}{\varnothing}
\renewcommand{\plog}{p_{\sf log}}

\usepackage{tabularx}

\newcommand{\plc}{{\sf PLC}}
 
\newcommand{\tsim}{t_{\sf sim}}

\newcommand{\tprep}{T_{\sf init}}
\newcommand{\tinit}{T_{\sf init}}

\newcommand{\off}{{\sf off}}
\newcommand{\rlog}{{\rho_{\sf log}}}

\newcommand{\polylog}{{\rm polylog}}

\newcommand{\kmax}{{k_{\sf max}}}

\usepackage{pifont}
\usepackage{ulem}
\newcommand{\cmrk}{{\color{Green}{\cmark}}}
\newcommand{\xmrk}{{\color{red}{\xmark}}} 
\renewcommand{\plc}{{\sf PLF}}
\newcommand{\plf}{{\sf PLF}}
 
\newcommand{\bw}{{\sf bw}}
\begin{document}

	\title{Local active error correction from simulated confinement}
	
	\author{Ethan Lake}	
	\email{elake@berkeley.edu} 
	\affiliation{Department of Physics, University of California Berkeley}

	\begin{abstract} 
		We refine an old idea for performing fault-tolerant error correction in topological codes by simulating confining interactions between excitations. We implement confinement using an array of local classical processors that measure syndromes, broadcast messages to neighboring processors, and move excitations using received messages. The dynamics of the resulting real-time decoder is geometrically local, homogeneous in spacetime, and self-organized, operating without any form of global control. We prove that below a threshold error rate, it achieves a memory lifetime scaling as a stretched exponential in the linear system size $L$, provided that it has access to $O(\polylog(L))$ noiseless classical bits for each noisy qubit. When applied to the surface code subject to depolarizing noise and measurement errors of equal strength, numerics indicate a threshold at $p_c \approx 1.5\%$.  
	\end{abstract}

	\maketitle
	\tableofcontents

	\section{Introduction and summary} \label{sec:intro} 
	
	Understanding which physical systems are capable of performing quantum error correction (QEC) is a central problem in quantum science. 
	From a theoretical perspective, the characterization of many-body systems capable of autonomously performing QEC is a key part of the ongoing effort to classify noise-robust quantum phases of matter (see e.g. \cite{coser2019classification,rakovszky2024defining,sang2024mixed,Dennis_2002,fan2024diagnostics,cubitt2015stability}).
	From a practical perspective, the construction of fast and accurate QEC schemes (``decoders'') is an essential step towards developing useful quantum computers: while significant progress has been made in recent years \cite{demarti2024decoding,campbell2024series,bluvstein2024logical,preskill2025beyond,bluvstein2025architectural,google_qec}, the best way of performing real-time decoding in a way faithful to the demanding data processing requirements and hardware constraints of modern devices is still the subject of much ongoing research (e.g. \cite{reilly2019challenges,brennan2025classical,maurya2024managing,ziad2024local,liyanage2024fpga,wu2023fusion,caune2024demonstrating,higgott2025sparse,zhou2025low}). 
	
	\begin{figure*}
		\centering 
		\includegraphics[width=1.05\tw]{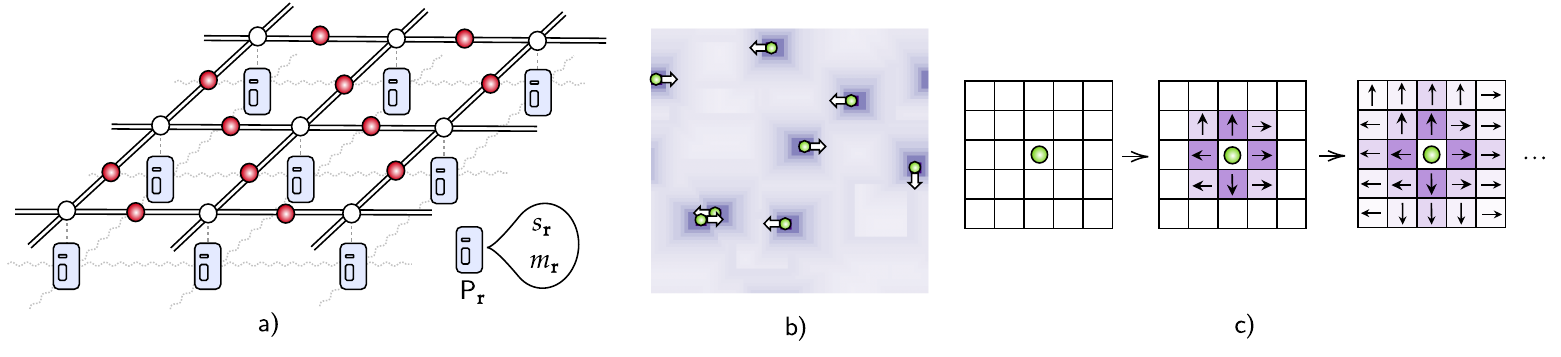}
		\caption{\label{fig:ccqs_fig} A summary of the message-passing architecture that simulates confinement. $\sfa)$ A schematic of a $2d$ CA decoder. Qubits (red balls) live on the links of a square lattice, and each lattice site $\bfr$ hosts a classical processor $\sfP_\bfr$. The classical system performs error correction through a combination of local measurements, local feedback, and communication between neighboring processors. In our decoder, each $\sfP_\bfr$ stores information about defects at $\bfr$ ($s_\bfr$) and messages that are used to generate confining interactions between defects ($m_\bfr$). $\sfb)$ The exchange of messages between defects (green balls) leads to confining interactions. The strength of the messages emitted by defects is indicated by the purple shading, and feedback is applied to move defects in the direction of the strongest received signal (white arrows).  ${\sf c)}$ A schematic of how defects produce messages. In $d$ dimensions, each defect produces multiple types of messages, one for each (positive and negative) coordinate direction. Messages emitted from $\bfr$ along the $\pm \uva$ direction spread to sites $\bfr'$ such that $r'^a-r^a \gtreqless r'^b- r^b$ for all $b=1,\dots,d$, and decrease their strength as they do so, with stronger messages overwriting weaker ones (and with degeneracies along the diagonals being lifted according to the convention in the figure).  }
	\end{figure*}
	
	This work studies the problem of performing real-time decoding on topological stabilizer codes in a way where all operations---both classical and quantum---are geometrically local. Decoders that operate in this setting may be constructed by placing a small classical processor at the location of each stabilizer generator, with each processor using a small amount of working memory to locally measure syndromes, communicate measurement results to its neighbors, and locally apply feedback (see fig.~\ref{fig:ccqs_fig} $\sfa$). These types of decoders are known as {\it cellular automaton} (CA) decoders \cite{Harrington2004,herold2015cellular,herold2017cellular,balasubramanian2024local,kubica2019cellular,breuckmann2016local,ray2024protecting,paletta2025high,lake2025fast}. The high degree of parallelization they afford is particularly relevant in systems where fast gate speeds mandate that error correction be performed quickly.

	We will be concerned in particular with CA decoders which operate without any form of global control, i.e. without any external ``overseer'' that supervises the error correction process, determining when and where various operations should be applied. In this setting, the processors must cooperatively gather local  information and make collective, distributed decisions about large-scale operations: the error-correction process must {\it self-organize}.

	In more detail, this work studies decoders satisfying the following conditions: 
	\begin{enumerate}
		\item {\it Locality:} All operations are local in spacetime, with classical data processing, quantum operations, and noise all occurring on a common timescale. 
		\item {\it Asynchronicity:} error-correcting operations do not have to be precisely synchronized between well-separated spacetime points.
		\item {\it Homogeneity:} the same operations are performed by the decoder at each point in space and time. 
	\end{enumerate}
	Let us elaborate slightly on these criteria. 
	Locality is a natural requirement to demand from a many-body physics perspective (when we say that classical and quantum operations occur on a common time scale, we mean that the decoder can accommodate a situation where the ratio of the classical and quantum clock speeds, with noise operating on the later time scale, is independent of the linear system size $L$). While local classical processing is not per se an important constraint in hardware design (in that near-term devices will not be limited by the time taken for light to cross the device), the high degree of parallelization it naturally produces is. 
	The asynchronicity condition is similar in spirit: requiring precise synchronization mandates that each processor have access to a shared perfectly reliable clock, which is a form of global control that cannot be present in any truly self-organized QEC scheme.\footnote{While synchronization is likely not an important issue in practice, it can become relevant at system sizes which are not ridiculously large (as an example, the FPGA-based surface code decoder of ref.~\cite{liyanage2024fpga} was estimated to run into a bottleneck from the global clock speed at $L \sim 50$).} 
	Finally, homogeneity---especially in time---is a natural property to demand from a many-body system that constitutes a nontrivial phase of matter. From an experimental point of view, its main benefit is that it simplifies the QEC architecture. 
	
	In addition to conditions 1-3, we would ideally also like the processors at each site to use only a constant ($L$-independent) number of classical bits for each qubit, and (at least from a theoretical perspective) we would like to allow these bits to themselves be subject to noise. Unfortunately, we will not be able to achieve this in the present work, and our analytic results will require that each processor store $O(\polylog(L))$ noiseless classical bits per qubit in order to achieve maximal error suppression (numerics indicate that $O(\log(L))$ may be sufficient). This is less than existing strategies based on windowed decoding (e.g. \cite{higgott2025sparse,chan2023actis,bombin2023modular,skoric2023parallel,tan2023scalable,wu2023fusion}), which require $\O(L)$ bits per processor, but is nevertheless conceptually rather unsavory. Whether or not this shortcoming can be overcome in a way that does not involve recourse to the notoriously opaque self-simulation methods of Gacs \cite{gacs2001reliable} is an important question to address in future work. 
	
	{\it Simulated confinement:} For topological codes, a particularly natural QEC scheme with the potential to satisfy criteria 1-3 is to correct errors by simulating a confining attractive interaction between anyons. At a high level, such an interaction is envisaged to move anyons towards their nearest neighbors, thereby inducing a correction that---at small enough error rates---is hoped to serve as a good proxy for a minimal weight matching. 
	This idea has been investigated at various points in the past \cite{Dennis_2002,hamma2009toric,pedrocchi2013enhanced,fujii2014measurement,herold2015cellular,herold2017cellular},\footnote{See in particular the last section of  \href{https://scgp.stonybrook.edu/video_portal/video.php?id=335}{this video lecture} by David Poulin for a nice exposition of  the simulated confinement idea.} but has yet to produce a real-time decoder with a threshold whose existence can be rigorously demonstrated. The constructions coming closest are the ``field-based decoders'' of \cite{herold2015cellular,herold2017cellular}, which use a classical scalar field obeying an anyon-sourced Poisson equation to mediate the desired attractive interaction. Unfortunately, these decoders do not meet our definition of locality, and in fact we will see that they do not possess a threshold in the fault-tolerant setting.

	The main contribution of this work is to show how a confining interaction can be simulated in a way that provably has a threshold, and satisfies criteria 1-3 above. Our construction uses an extension of the message-passing architecture introduced by the author in ref.~\cite{lake2025fast}: in the CA decoder of that work, processors send out messages broadcasting the locations of anyons, and feedback is applied to move anyons in the direction of the messages they receive. This creates a confining interaction between anyons in space, and was shown to produce a threshold for offline decoding.\footnote{By ``offline decoding'', we mean the problem of decoding noisy input states under the assumption that the decoding process is noiseless and measurements are error-free. This is also often referred to as the ``code capacity scenario''.} In this work, we show that this alone is not sufficient to produce a threshold when transient noise and measurement errors are present: in this setting, neither this construction nor the field-based decoders of \cite{herold2015cellular,herold2017cellular} possess a threshold (although the error suppression they achieve may still be useful in an experimental context). When measurements are faulty, the relevant objects that must be annihilated during error correction are not anyons, but rather ``defects''---spacetime events where a syndrome changes its value between measurements \cite{Dennis_2002}---and instead of a confining interaction between anyons in space, we will see that what is needed is a confining interaction between {\it defects} in space{\it time}. We will show that this can be implemented by the addition of a ``small'' (at most $O(\polylog(L))$) number of extra classical bits to each processor, and will rigorously prove that this restores the presence of a threshold. 
	
	\ss{Extended summary}
	
	We now provide a summary of our main results, which will also serve as an outline of the paper. 
	
	\begin{figure*}
		\centering 
		\includegraphics[width=.9\tw]{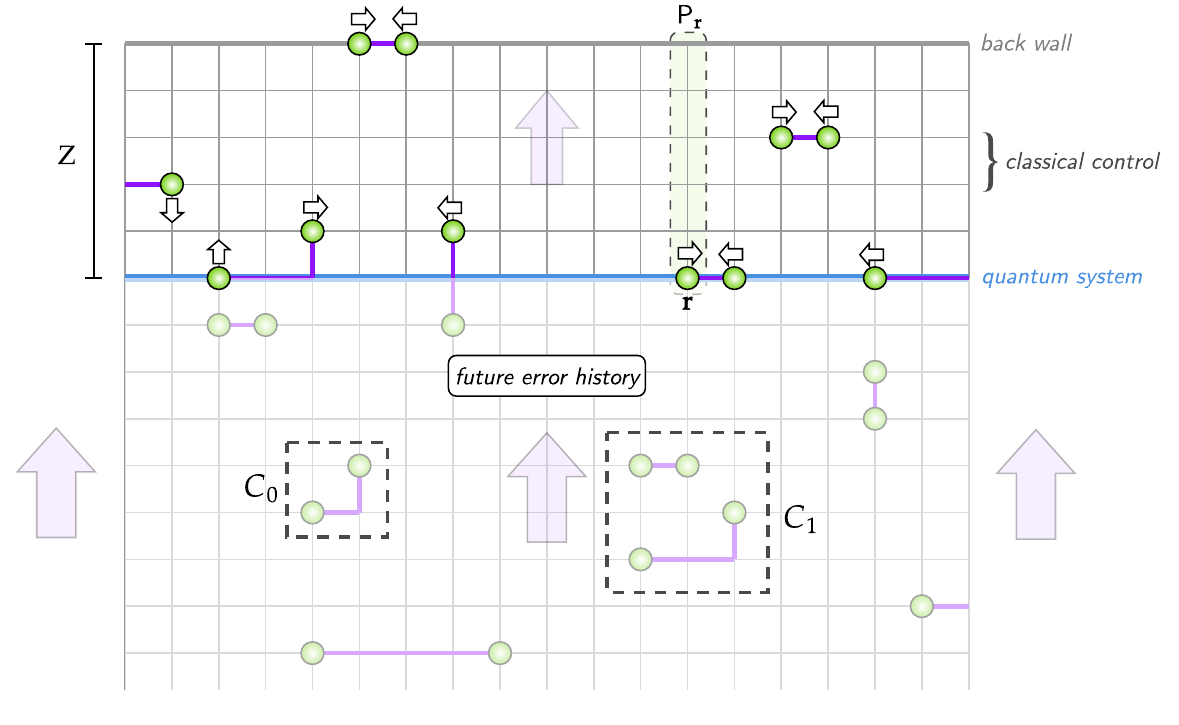}
		\caption{\label{fig:arch_schem} A schematic of the decoder operation and error-correction process, shown in one spatial dimension for simplicity of illustration. Qubits live on the links of the $1d$ chain indicated in blue, and the lattice above this chain (indicated as the {\it {\sf classical control}} region) contains the internal variables stored by the classical processors $\sfP_\bfr$ at each site $\bfr$ (an example of the internal variables controlled by a particular processor is indicated by the green shaded region). The processors store a lattice of classical variables of depth $Z$, and we prove that a threshold is present as long as $Z = \O(\polylog(L))$.  Syndrome-changing events (``defects'') are marked as green circles, and move ``upwards'' in the classical control region under the decoding dynamics until accumulating on the {\it {\sf back wall}}. During this process, defects use the message-passing scheme to attempt to pair-annihilate; the white arrows indicate the directions that the anyons would move in the limit of instantaneous message speeds. 
			The section labeled {\it {\sf future error history}} illustrates errors in spacetime that have not yet been experienced by the system; as time increases these errors are ``fed into'' the quantum system, and then propagate up along the $Z$ direction. The regions marked $C_0$ and $C_1$ indicate noise clusters of varying spacetime support. We show that only clusters of linear size $\gtrsim Z$ can survive to the back wall, where they are then corrected unless they coagulate into a larger cluster of size $\sim L$. }
	\end{figure*}
	
	\sss{Simulating confinement with message passing}
	
	We begin in sec.~\ref{sec:setup} by describing the architecture used by our decoder. Let us first recall the message-passing setup of ref.~\cite{lake2025fast}. At the location $\bfr$ of each stabilizer generator, we place a classical processor $\sfP_\bfr$, which stores two types of information. The first is a copy of the most recently-measured syndrome value $\s_\bfr$ at $\bfr$. The second type is a collection of integer-valued messages labeled by tuples $(s,a)$, with $s\in\{\pm1\}, a \in\{1,\dots,d\}$ with $d$ the spatial dimension (we will consider both $2d$ topological stabilizer codes and the $1d$ repetition code in this work). Anyons (viz. sites with $\s_\bfr = -1$) are taken to emit integer-valued messages, with a message labeled by $(s,a)$ propagating along the $s\uva$ direction and increasing its value at each time step (with smaller-valued messages overwriting larger-valued ones). Feedback is applied by moving each anyon in the direction of the smallest-valued message it receives, with the resulting dynamics serving as a local, time-delayed way of moving each anyon in the direction of its nearest neighbor (see fig.~\ref{fig:ccqs_fig}~$\sfb,\sfc$). 
	
	This setup lets us simulate a confining interaction between anyons in space, which as shown in \cite{lake2025fast} is sufficient for performing offline decoding. As mentioned above, what is needed to perform real-time decoding is rather to simulate confining interactions between defects in spacetime. The most naive way of doing this is to store a full spacetime history of the defects in the $\sfP_\bfr$, and to then perform the aforementioned message-passing decoding dynamics on this history. While parallel windowing techniques could be used to reduce the temporal depth of histories that need to be stored to $\ct(L)$, this approach would nevertheless incur a significant classical overhead, and (more importantly for us) would not satisfy criteria 1-3 above. 
	
	Instead of storing large chunks of the defect spacetime history, our solution will be to store a dynamically-updated buffer of the history on which error correction is constantly being performed. As soon as defects are produced, they enter the buffer and begin using message-passing dynamics to correct themselves. This is done by enlarging each $\sfP_\bfr$ to store defects at $Z$ additional classical locations, which we may arrange most naturally along an additional spatial dimension of size $Z$ (see fig.~\ref{fig:arch_schem}). 
	As defects are produced by noise, the processors move them ``up'' along this direction, and they undergo $(d+1)$-dimensional message-passing dynamics along the way. The message-passing ensures that defects begin to be corrected as they move along the $z$ dimension, with the defects becoming sparser as $z$ gets larger. Intuitively speaking, the flow of defects along the $z$ direction should be thought of as simulating a form of RG flow on the spacetime defect history. 
	
	The surface at $z=Z$---which we will refer to as the ``back wall''---is used to collect defects associated with the largest-scale errors. When a defect reaches $z=Z$ it ``sticks'' to the back wall, where it subsequently undergoes $d$-dimensional message-passing dynamics until it meets another back-wall defect and is annihilated. The dynamics on the back wall can thus be thought of as running a decoder with $Z=0$ subjected to a renormalized error rate, a correspondence that will be utilized to prove the existence of a threshold provided that a sufficient amount of RG flow occurs (which for us will dictate taking $Z = O(\polylog(L))$). See fig.~\ref{fig:arch_schem} for an illustration of these points.

	\sss{Screening and pseudothresholds}
	
	We will characterize the performance of a decoder using its {\it memory time} $\tmem$, defined as the expected time the decoder can be run before suffering a logical error. A decoder is said to have a {\it threshold} at $p_c$ if for noise of strength $p <p_c$, the memory time diverges as the system size is increased: 
	\be \lim_{L\ra\infty} \tmem = \infty \, \, \forall \, \, p < p_c.\ee 
	
	In sec.~\ref{sec:nogo}, we investigate the scaling of $\tmem$ when we perform active error correction with $Z=0$, viz. without any RG. In this limit---which is equivalent to the decoder studied in \cite{lake2025fast}---the confining interaction extends only in space.
	We argue that in this limit, our decoder, as well as the field-based decoders of \cite{herold2015cellular,herold2017cellular}, do {\it not} have a threshold in the presence of transient noise, even when stabilizer measurements are perfect.\footnote{For global decoders, active (real-time) decoding and offline decoding are closely related in the absence of measurement errors. For local decoders this need not be the case, and indeed for us measurement errors are not the essential issue (rather, transient errors are).} Instead, the memory time behaves as 
	\be \label{pseudothresh_intro} \tmem \leq  \exp(a\min(L,b/p^{1/d})) + c,\ee 
	for some $O(1)$ constants $a,b,c$, which is finite (but can be quite large at small $p$) when $L\ra\infty$. 
	The reason for this scaling is quite simple, and occurs when the signals produced by noise ``screen'' the confining interactions between well-separated anyon pairs. 
	We rigorously show that \eqref{pseudothresh_intro} holds in a particular $p$-bounded error model, and show numerically that it holds for i.i.d noise as well. 
	We refer to a decoder with this behavior as having a {\it pseudothreshold}. Pseudothresholds can mimic thresholds in numerics at small system sizes, as $\tmem$ increases superpolynomially in $L$ until a scale $L_*\sim 1/p^{1/d}$.\footnote{With the association $1/p \sim e^{\b}$, the scalings  $\tmem \sim e^{e^\b}$ and $L_*\sim e^\b$ match those of several ``partially self-correcting memories'' realized in disordered $3d$ quantum codes \cite{michnicki20143d,siva2017topological,williamson2023layer,gu2025layer}. }

	\sss{Decoding performance}
	
	While the decoder with $Z=0$ does not have a threshold, $\tmem$ diverges very quickly as $p\ra0$, and increases with $L$ up to a system size that also diverges as $p\ra0$. These decoders are thus ``close'' to having a threshold, and in sec.~\ref{sec:results} we prove that a threshold can be restored---and maximal error suppression achieved---by taking $Z=O(\polylog(L))$. A polylog amount of RG is thus sufficient to filter the noise distribution to the point where the offline decoder on the back wall can achieve a threshold. 
	We formalize this result as follows: 
	\begin{theorem}[existence of a threshold, informal] 
		Consider the message-passing decoder running on a system of linear size $L$ and buffer depth $Z$. For any error model of strength $p$, there exist positive constants $a,\z,\b,p_c$ such that as long as $p\leq p_c$ and 
		\be Z \geq a \(\frac{\log L}{\log(p_c/p)}\)^{1/\z} ,\ee 
		then the memory lifetime satisfies 
		\be \label{tmem_intro} \tmem = (p_c/p)^{\O(L^\b)}.\ee 		
	\end{theorem}
	While our analytic results require $Z = O(\polylog(L))$, we believe that a threshold should be present for $Z = O({\rm poly}(\log\log L))$ (which however would give a reduced scaling of $\tmem$ below threshold). 
	
	The proof strategy is fairly simple. It uses a spacetime version of the linear erosion property proved for the message-passing decoders in \cite{lake2025fast}, and combines this with a standard clustering argument \cite{gacs2001reliable,bravyi2011analytic} to show that the decoding dynamics corrects clusters of errors of size $\lesssim Z$ before they reach the back wall. The clustering argument then shows that the back wall experiences a noise model with effective error rate of $p_Z \sim p^{Z^\z}$. Since the lifetime of the decoder with $Z=0$ is $\tmem \sim e^{b/p_Z^{1/d}}$, we have $\tmem \sim e^{bp^{-Z^\a/d}}$. As this is doubly exponential in $Z$, we need only take $Z = \polylog(L)$ to guarantee a superpolynomially long lifetime.
	
	To benchmark our decoder's performance, we numerically estimate the threshold error rate it attains for the $2d$ surface code subjected to i.i.d depolarizing and measurement noise of equal strength. For synchronous updates, we find $p_c \approx 1.5\%$, about half the threshold achieved by running minimal-weight perfect matching on the entire spacetime history \cite{wang2003confinement} (for the simplest continuous-time asynchronous scheme we consider, $p_c$ is reduced to $\approx 0.5\%$). For a fully local decoder this is actually quite respectable, especially since no effort has been made to optimize this number; the only other fully local decoder with a rigorously provable threshold we are aware of is \cite{balasubramanian2024local}, where $p_c \lesssim 0.01\%$.

	
	Local decoding algorithms must necessarily be  greedier than global ones, and we show that this greediness comes at a cost: in sec.~\ref{ss:gerry} we prove that the constant $\b$ in \eqref{tmem_intro} is {\it strictly} less than 1, due to a phenomenon reminiscent of Gerrymandering. This is proved by using a result from \cite{lake2025fast} showing that we may always create a Cantor-set-like pattern of errors which has weight $o(L)$ and is guaranteed to produce a logical error when decoded, and then using another noise-clustering argument to show that the entropy of noise events containing such patterns is non-negligible.

	\ss{Comparison with prior work}

	\begin{table*}[ht]
		\centering
		\renewcommand{\arraystretch}{1.5}
		\begin{tabular}{lccccc}
			
			\textbf{} &  \,{\shortstack[c]{\textbf{windowed} \\ \textbf{MWPM} \cite{skoric2023parallel,tan2023scalable,bombin2023modular}}}\,& \, \textbf{Harrington}~\cite{Harrington2004} \, & \, \textbf{$2d$ model of ref.~\cite{balasubramanian2024local}} \, &  \, \textbf{field-based}~\cite{herold2015cellular,herold2017cellular} \, & \, \textbf{this work}\, \\
			\hline\hline
			local & \xmrk & \cmrk & \cmrk & \xmrk & \cmrk \\
			\hline
			homogeneous & \xmrk & \xmrk & \xmrk & \cmrk & \cmrk \\
			\hline
			asynchronous & \xmrk & \xmrk & \xmrk & \cmrk & \cmrk 
			\\ 
			\hline
			best $\tmem$ & $(p_c/p)^{\ct(L)}$ & $(p_c/p)^{\ct(L^\b)}$ & $(p_c/p)^{\ct(L^\b)}$ & $\ct(L^0)$ & $(p_c/p)^{\ct(L^\b)}$ \\  \hline 
			classical overhead & $\o(L)$ & $\polylog(L)$ & $\log(L)$ & $\o(L)$ &  $\polylog(L)$ \\ \hline 
			pheno. $p_c$ & $2.9\%$ \cite{wang2003confinement} & $0.1\%$\,?  \cite{breuckmann2016local} & $\lesssim 0.01\%$ \cite{balasubramanian2024local} & n/a& $1.5\%$ 
		\end{tabular}
		\medskip 
		\caption{\label{tab:comparison} Comparison of various decoders for the $2d$ surface code. ``best $\tmem$'' indicates the optimal scaling of the memory lifetime with system size $L$ under a $p$-bounded noise model (with $\b$ a positive constant strictly smaller than 1), and ``classical overhead'' denotes the number of noiseless classical bits used by the decoder for each noisy qubit in order to achieve this scaling (the majority of which may be inactive during typical sub-threshold noise events). For the decoder of ref.~\cite{balasubramanian2024local} the classical bits are read-only, and can be eliminated entirely in a $3d$ geometry stabilizing ${\rm poly}(L)$ planes of $2d$ toric codes. ``pheno $p_c$'' denotes the approximate error threshold under a phenomenological noise model where i.i.d depolarizing noise and measurement errors are applied with equal strength. Harrington's decoder has what might be a threshold near $0.1\%$, although both the theory and numerics are difficult to conclusively interpret. The field-based decoders do not have a threshold, but display pseudothreshold behavior that mimics a threshold near $0.3\%$ at small system sizes \cite{herold2017cellular}. For our work, we have reported the threshold obtained under synchronous updates. 
		}
	\end{table*}
	
	The original demonstration of fault-tolerant error correction was made in the context of concatenated codes, and was shown to satisfy criterion 1 \cite{aharonov1996faulttolerantquantumcomputation,gottesman2000fault}, but not 2 or 3. More relevant to this work are decoders for $2d$ topological codes, of which many variants exist. 
	A table comparing our work to a selected subset of constructions in the literature is given in tab.~\ref{tab:comparison}, on which we elaborate more below. 
	
	\begin{itemize}			
		\item {\it Windowed decoding:} This approach saves a spacetime history of defect locations to memory, and then (nonlocally) performs global decoding on the recorded history \cite{Dennis_2002}. This idea can be made more efficient by chunking the history into windows which are processed in parallel \cite{skoric2023parallel,tan2023scalable,bombin2023modular}; to achieve maximal error suppression these windows must have a temporal extent of $\O(L)$, for which they achieve essentially the same accuracy as the global decoder acting on the entire history, and decode at a speed determined by that of the global decoder on a spacetime box of size $L^3$. 
		
		In some sense, our approach can be viewed as a dynamically-adapted windowing scheme that is fully translation invariant in space and time. Rather than  having windows hard-coded into the decoding architecture, the RG structure present in the message-passing dynamics autonomously recognizes when hard-to-correct errors have occurred, and defers their correction to sites with larger $z$ coordinates. 		
		
		\item {\it Hierarchical approaches:} The first type of hierarchical decoders are those that act nonlocally in spacetime. They  partition the spacetime history of defects into a hierarchical grouping of boxes, which are then decoded starting from small scales and recursively moving on to larger ones \cite{bravyi2011analytic,duclos2013fault,wootton2015simple,sang2024mixed}. Since they act nonlocally they can obtain fairly high values of $p_c$;\footnote{When we refer here and below to a decoder as having a threshold at $p_c$, we are refering to the threshold obtained for the $2d$ toric code subjected to i.i.d depolarizing noise and measurement noise of equal strength.}  as an example, the protocol of \cite{duclos2013fault} has a threshold at $p_c\approx 1.9\%$. 
		
		The second type are hierarchical CA decoders (the first such construction being Harrington's decoder \cite{Harrington2004}), which perform a similar type of boxing, but do so locally in spacetime. These decoders violate conditions 2 and 3,\footnote{Strictly speaking, the dynamics can be made time- and space-translation invariant by embedding a (large) set of instructions into each processor, which is then advanced in a translation-invariant way. This is a bit against the spirit of what we mean by ``homogeneous'', and we will choose to formally rule it out by demanding that the decoder also function properly when initialized on homogeneous input states.} although the decoder of \cite{balasubramanian2024local} can be made time-translation invariant in three dimensions. Empirically, these decoders are observed to sacrifice a large amount of error-correcting power to achieve locality: Harrington's decoder may have a threshold of at most $0.1\%$ \cite{breuckmann2016local} (although the evidence from both analytics and numerics is rather unclear), and while the decoder of \cite{balasubramanian2024local} (which has less classical overhead than Harrington's decoder) was rigorously proven to have a threshold, the value of $p_c$ was too small to be accurately estimated (part of the difficultly with numerically studying these models is that they are only naturally defined on systems of linear size $n^l$ where $n,l$ are integers, limiting the number of computationally-accessible distinct system sizes). 
		
		Roughly speaking, these constructions may be viewed as adopting a block-spin approach to RG, committing to a particular pattern of coarse-graining in advance. 
		While the $z$ dimension in our approach is also used to implement a form of RG, its interpretation is closer to strong disorder RG: it adaptively selects out and corrects clusters of errors using the heirarchy of scales already present in the noise, rather than hard-coding an explicit hierarchy into the decoding dynamics.

		\item {\it Field-based decoders:} The decoders that served as the biggest inspiration to the present work are the ``field-based decoders'' of refs.~\cite{herold2015cellular,herold2017cellular}, which use a local cellular automaton to induce an attractive long-range force between anyons. To avoid problems stemming from the self-interaction between an anyon and the force field it itself generates, it was shown in \cite{herold2015cellular} that the speed of classical processing must diverge as $L\ra\infty$ (as such, these decoders do not obey our definition of locality). Furthermore, in order to generate a force profile that falls off quickly enough that signals from distant anyons do not swamp those from nearby ones, the classical processors must live in at least three spatial dimensions. 
		
		Our message-passing framework solves the self-interaction problem by having messages ballistically spread outwards from defects, so that defects can never ``catch up'' to the signals they emit. It also does not require the addition of an extra dimension to reduce the strength of the force. More significantly, as mentioned above, field-based decoders are not fully fault tolerant due to screening effects, and the additional $O(\polylog(L))$ classical bits used by our decoder to store the dynamically-updated defect history are needed to fix this problem. 
		
		\item {\it Energy-based confinement:} A final class of decoders are those based on performing thermal dynamics under a Hamiltonian with long-range interactions that energetically confine anyons. These approaches either involve Hamiltonian terms whose strength diverges with system size, or require coupling to a bosonic field that relaxes thermodynamically fast \cite{hamma2009toric,pedrocchi2013enhanced,chesi2010self} (see also \cite{landon2015perturbative,brown2016quantum}). Most importantly, they only correct errors when the noise is thermal with respect to the Hamiltonian in question, and are unable to deal with e.g. even very weak depolarizing noise (we by contrast aim to design decoders that resit {\it any} sufficiently weak noise process).\footnote{For example, error correction in the model of ref.~\cite{pedrocchi2013enhanced} is achieved by virtue of anyon pairs having thermodynamically divergent energy, so that detailed-balance-obeying thermal noise is unable to create anyons. Weak depolarizing noise is ``infinite temperature'' in this respect, and cannot be dealt with energetically.}
	\end{itemize}
	
	\section{Preliminaries}\label{sec:setup}

	\ss{Message-passing architecture} \label{ss:setup}
	
	We begin by defining the message-passing architecture employed by our decoder, which extends the construction of \cite{lake2025fast} in a way that will be essential for performing active decoding. We will focus on the case of a model with $\zt$ excitations for simplicity (viz. either the $1d$ repetition code or the $2d$ toric code), with the qubits laid out in an $L^d$ grid square grid with periodic boundary conditions (although the general construction works for any topological stabilizer code and on open boundary conditions). 
	
	As illustrated in fig.~\ref{fig:arch_schem},  the classical control system that performs error correction is most naturally arranged on an $L^d \times Z$ lattice, where $Z$ is the depth of the control system along a (small) extra dimension; as mentioned above, to get the longest possible memory lifetime, we will never need to let $Z$ scale faster than $\polylog(L)$ (smaller values of $Z$ can still produce a threshold, but may yield smaller values of $\tmem$). We will denote sites of this lattice as 
	\be \bfx = (\bfr,z),\qq \bfr = (r^1,\dots,r^d) \in \zz_L^d,\ee 
	and will use the notation $x^i = r^i$ for $1\leq i \leq d$, and $x^{d+1} = z$. 
	
	The classical control system stores various types of information at each site, which we compartmentalize as $s_\bfx \in \{\pm1\}$ and $m_\bfx^{\pm a} \in \{\infty,1,\dots,m_{\sf max}\}$, with $a \in \{1,\dots,d+1\}$ (the reason for using ``$\infty$'' instead of ``$0$'' will become clear shortly). Sites with $s_\bfx = -1$ will be referred to as hosting ``defects''. The $m^{\pm a}_\bfx$ variables will store messages used to communicate between defects at different locations. $m^{\pm a}_\bfx$ messages propagate along the $\pm \uva$ direction, with their values equal to the time that they have been propagating (with a value of $\infty$ meaning that no message is present). For simplicity we will set $m_{\sf max} = L$, so that storing the values of $m_\bfx^{\pm a}$ requires at most $\lceil \log_2(L) \rceil$ bits. 
	We believe it should be possible to take $m_{\sf max} = O(L^0)$ without compromising the performance of the decoder (in terms of the scaling of $\tmem$ with respect to $p$ and $L$), by taking the messages to evolve according to the proposal in appendix E of \cite{lake2025fast} (which unfortunately is rather difficult to analyze analytically). Because of this, we will not worry about explicitly addressing how to connect the bits associated with storing the values of $m_\bfx^{\pm a}$ in a strictly geometrically local way. 
	
	For simplicity, we will first describe the operation of the decoder in discrete time, with updates synchronized across all sites; various ways of desynchronizing the dynamics will be discussed in sec.~\ref{ss:lind}. At each time step $t$, the automaton does the following in sequence:
	\begin{enumerate}
		\item Measures syndromes at each site $\bfr$, obtaining (potentially faulty) outcomes $\s_\bfr(t)$. 
		\item For each $\bfr$, checks if $\s_\bfr(t) = \s_\bfr(t-1)$,\footnote{Only the two most recently values of $\s_\bfr(t)$ are kept track of.} and if not, creates an defect at  $(\bfr,0)$ by setting $s_{(\bfr,0)} = -1$. 
		\item Moves all defects and messages ``up'' along the $z$ direction, with defects being absorbed onto the back wall (the surface where $z=Z$) when they reach it: 
		\bea \label{rg_cycle} s_\bfx & \mt \begin{dcases} s_\bfx \cdot s_{\bfx-\uvz} & x^{d+1} = Z \\ s_{\bfx-\uvz} & 1 < x^{d+1} < Z \end{dcases} \\  
		m_\bfx^{\pm a} &\mt m^{\pm a}_{\bfx-\uvz} \qq 1<x^{d+1}<Z. \eea 
		\item Performs one round of message-passing dynamics and feedback (see below).
	\end{enumerate}
	
	Without the final step, this would be equivalent to storing a running spacetime syndrome history over a temporal ``buffer window'' of depth $Z$, with syndrome-changing events being deposited onto the back wall.		
	The message-passing dynamics in the last step is responsible for performing error correction, and is only a slight modification of the construction introduced in \cite{lake2025fast}. Intuitively, this scheme attempts to locally compute a proxy for a minimal weight perfect matching (MWPM) of the spacetime error syndromes, with the proxy dynamically updated as new measurements are acquired. Small errors with spatial separation less than $Z$ are corrected before they reach the back wall, while the back wall stores (and eventually corrects) large errors that require a time greater than $Z$ to be matched. 
	
	A single round of message-passing dynamics in the bulk $(1<x^{d+1}<Z)$ consists of the following steps: 
	\begin{enumerate}
		\item Each defect sources messages, which is done by setting $m^{\pm a}_{\bfx\pm\uva} = 1$ at all locations where $s_{\bfx\mp \uva} =-1$. 
		\item Messages of type $\pm a$ propagate along the $\pm\uva$ direction, and create copies of themselves on the neighboring sites in the $\pm\uvb$ directions, with $a\neq b$. Each time they propagate to a new site, they increase their value to reflect the distance they have been propagating.\footnote{With a smaller value of $m_{\sf max}$, we could take e.g. the messages to reset to $\infty$ if they reach $m_{\sf max}$, or to remain at $m_{\sf max}$ but continue to propagate. Since we will set $m_{\sf max} = L$ in the text, we will not worry about this choice.} Finally, if a message of type $\pm a$ propagates to a site where another message is already present, the smaller of the two messages is taken: ``newer'' messages thus overwrite ``older'' ones. 
		
		To describe this formally, define the region 
		\be \label{cdef} \sfC^{\pm a}_\bfx= B_{\bfx}^{(\infty)} \cap B_{\bfx\mp 2\uva}^{(\infty)},\ee 
		where $B_{\bfx}^{(\infty)} = \{ \bfx' \, : \, ||\bfx - \bfx'||_\infty \leq 1\}$ is the unit $\infty$-ball centered at $\bfx$. We then update the messages in the bulk as\footnote{We have chosen to use the 1-norm in the above equation only because it produces higher threshold values under i.i.d noise (as compared to the $\infty$ norm). }
		\be \label{mupdate} m^{\pm a}_\bfx \mt \min \{ m^{\pm a}_{\bfx'} + ||\bfx' - \bfx||_1 \, : \, \bfx' \in \sfC^{\pm a}_\bfx\}.  \ee

		\item The previous two steps are repeated a total of $v$ times, where $v$ is the ``velocity'' of the messages (attaining a threshold will be shown to require $v>2$). 
		
		\item A defect at $\bfx$ is moved in the direction of the smallest-valued message at its location, viz. it moves along $s\uva$, where\footnote{We resolve degeneracies between the different directions as in \cite{pajouheshgar2025exploring}.}
		\be (s,a) = {\rm argmin}_{s',a'} \{ m^{s' a'}_\bfx \}.\ee 
		When defects are moved along a spatial direction, an appropriate Pauli operator implementing their movement is applied to the system.\footnote{To do this in a way which is geometrically local in the $z$ dimension, $\zt$-valued signals storing the parity of the correction at a given link can be broadcast down from larger to smaller $z$ values. We will not worry about explicitly implementing this.} When anyons are moved along $\pm \uvz$ (which they will do when eliminating measurement errors), no such operator is applied.
	\end{enumerate} 
	
	Finally, we slightly modify the message-passing dynamics on the back wall, where it simply performs $d$-dimensional message passing, ignoring messages along the $\pm\uvz$ direction: $m^{\pm a}_\bfx$ are defined only with $a \in \{1,\dots,d\}$, and are updated as in \eqref{mupdate} but with $\sfC^{\pm a}_\bfx$ defined as in \eqref{cdef} but with $d$-dimensional balls. 
	
	In the above description of the automaton we have adopted periodic boundary conditions for simplicity (and will continue to do so in the remainder of this work). With open boundary conditions, it is sufficient to simply take messages to additionally be sourced from rough boundary sites (see \cite{lake2025fast} for details).

	\ss{Noise models and clustering} 
	
	As is common in the error correction literature, we will assume for simplicity that the noise our system experiences is described by an incoherent $p$-bounded process. In more detail, define a noise realization $\sfN$ to be a collection of spacetime points (including both qubit and measurement bit locations) where the noise applies a nontrivial Pauli operator or flips a measurement bit. The noise being $p$-bounded means that the marginal probability of any particular subset $A$ of spacetime locations being in $\sfN$ decays as 
	\be \sfP(A \subset \sfN) < p^{|A|}.\ee 
	This is essentially the most general incoherent noise model one could consider, and does not require that the noise have short range correlations in spacetime. While this definition does not allow for any coherence between errors, it turns out that coherent errors may be handled in an identical way, using the operator decomposition techniques of Gottesman (see \cite{aliferis2005quantum} and in particular chapter 15 of \cite{gottesman2024surviving}). The upshot is that all of our analytic results about thresholds in sec.~\ref{sec:results} hold for extremely general types of error models: roughly, any error model which couples the quantum system and measurement bits to a (perhaps infinitely large) environment with a time-dependent Hamiltonian, with the system-environment couplings being sufficiently weak, will work. Nevertheless, for simplicity, and to avoid a lengthy re-explanation of the techniques in \cite{aliferis2005quantum,gottesman2024surviving}, we will restrict to $p$-bounded stochastic noise channels throughout. 
	
	In the context of CSS codes it is always possible to break up the decoding task into separate tasks for $X$-type and $Z$-type anyons; in what follows we will assume that this has been done, and will focus only on one sector (our noise model is general enough to accommodate correlations between sectors, see e.g. \cite{balasubramanian2024local}). For simplicity, we will also assume that the anyons in this sector have $\zt$ fusion rules, as in the surface code. As explained in corollary~\ref{cor:generalization}, our construction generalizes immediately to any stabilizer code described by an Abelian anyon theory. 
	
	Like many other error correction schemes, the proof that our decoder has a threshold will be given by leveraging a clustering property possessed by $p$-bounded noise at small enough $p$ (see \cite{gacs2001reliable,ccapuni2021reliable,bravyi2011analytic}, etc.). To keep the discussion self-contained, we provide a brief review of the relevant notions below, taking the presentation almost verbatim from \cite{lake2025fast}.\footnote{The definitions used here (and in \cite{ccapuni2021reliable}) differ from the ``chunking'' definition in \cite{bravyi2011analytic}, but for our present application the difference is not important. }
	
	\ms\begin{definition}[clustering]\label{def:clustering}
		Let $\sfN$ be a noise realization. A $(W,B)${\it  -cluster} $C_{(W,B)}$ with size $W$ and buffer $B$ is a subset of $\sfN$ such that 
		\begin{enumerate}
			\item all points in $C_{(W,B)}$ are contained in a radius-$W/2$ $\infty$-ball $B_u(W/2)$ for some spacetime point $u=(\bfx,t)$, and 
			\item these points are separated from other points in $\sfN$ by a buffer region of thickness at least $B$:
			\be \sfN \cap (B_u(W/2+B) \setminus  B_u(W/2)) = \emp.\ee 			
		\end{enumerate}
		We call a spacetime point $u\in \sfN$ {\it $(W,B)$-clustered} if it is a member of a $(W,B)$ cluster, and say that a set is $(W,B)$-clustered if all points it contains are. 
	\end{definition}\ms
	
	We now use this definition to hierarchically refine $\sfN$ into clusters of exponentially increasing sizes: 
	\ms\begin{definition}[noise hierarchies and level-$k$ error rates]\label{def:noise_hier}
		Let $\sfN_0 = \sfN$, and fix positive constants $n,w,b$ such that $w<b$. 
		A $(wn^k,bn^k)_\infty$-cluster will be referred to as a $k$-cluster. 
		
		The {\it $(k+1)$th level clustered noise set $\sfN_{k+1}$} is defined by deleting all $(wn^k, bn^k)$-clustered points from $\sfN_k$.			
		For a $p$-bounded error model $\mce$, the {\it level-$k$ error rate} $p_k$ is the largest probability for any given spacetime location to belong to $\sfN_k$:
		\be p_k = \max_{u}\sfP_{\sfN \sim \mce}\(u\in \sfN_k\).\ee 
	\end{definition}\ms 
	
	The upshot of this clustering definition is that at small enough $p$, large clusters are very rare: the techniques of \cite{gacs2001reliable} show that as long as 
	$0 < w < b < \frac{n-3}4 w$, then there exists a positive constant $p_c$ such that \cite{lake2025fast}
	\be \label{rarenoise} p_k \leq  \(\frac p{p_c}\)^{2^k}.\ee

	\ss{Memory time and logical error rate}
	
	For a noiseless active decoder $\mcd$, we will write $\mcd_{\mce}$ to denote the implementation of $\mcd$ in the presence of a noise process $\mce$. The main quantity we will use to diagnose the performance of $\mcd_\mce$ is the memory time: 
	\begin{definition}[memory time]\label{def:tmem} 
		Given a noiseless offline decoder $\mcd_{\sf off}$,\footnote{viz. one which performs perfect measurements and uses noiseless data processing to noisy input states to noiseless logical ones.} we define the memory time $\tmem$ of $\mcd_{\mce}$ as the expected time at which offline decoding with $\mcd_{\sf off}$ fails:\footnote{Here $\mcd_{\mce}^t$ is shorthand for the channel obtained by evolving an input state for $t$ time steps under the noisy error correcting dynamics, and then tracing out the environment that $\mce$ couples the system to. If $\mce$ has temporal correlations, this need not be expressible as $t$ powers of a quantum channel acting on only the system and measurement qubits.}
		\be \label{tmemdef} \tmem = \max_{\rlog} \min \{ t \, : \, || \mcd_{\sf off} (\mcd_{\mce}^t(\rlog)) - \rlog||_1 < 1/2 \}.\ee 
		$\mcd$ will be said to have a {\it threshold} at $p_c$ under the noise $\mce$ if 
		\be \lim_{L\ra\infty} \tmem = \infty \, \, \forall \, p < p_c. \ee 
	\end{definition}
	
	The implicit dependence of $\tmem$ on $\mcd_{\sf off}$ is not so important for our purposes (provided at least that $\mcd_{\sf off}$ itself has a threshold). For the $2d$ toric code, we will take $\mcd_{\sf off}$ to be $\mcd$ run for a time sufficiently long to guarantee the removal of all defects (for us, this time is $O(L)$). For the $1d$ repetition code, we will simply use a global majority vote. 
	
	In numerics, we will also find it helpful to compute the {\it logical error rate} $\plog$, which we define as the probability of making a logical error after evolving for a time $L$ (as judged by an application of $\mcd_{\sf off}$): 
	\be \label{plogdef} \plog = \frac12 \max_{\rlog} || \mcd_{\sf off}(\mcd_{\sf on,\mce}^L) - \rlog ||_1 .\ee

	\section{Message-passing at $Z=0$ and field-based decoders: no-go results and pseudothresholds}\label{sec:nogo}


	\begin{figure*}
		\centering \includegraphics[width=.75\tw]{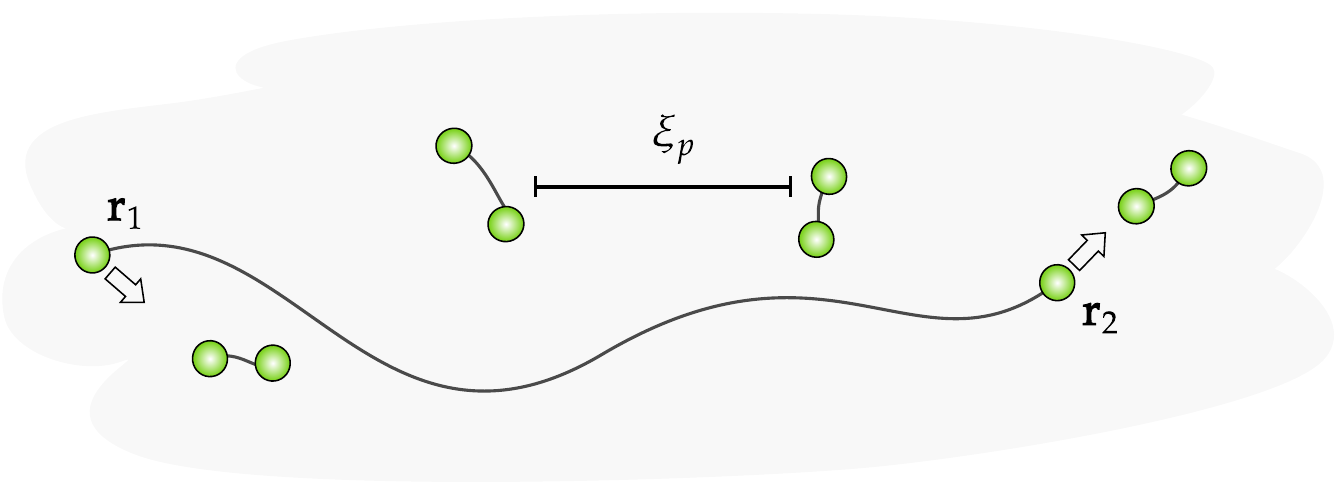}
		\caption{\label{fig:screening_fig} A schematic illustrating why $Z=0$ message-passing decoders and field-based decoders have a pseudothreshold. For noise of strength $p$, the typical distance between anyon pairs created by the noise is $\xi_p \sim p^{-1/d}$. When a pair of anyons $a_1,a_2$ at $\bfr_1,\bfr_2$ is created with separation $||\bfr_1-\bfr_2||\gg \xi_p$, the intervening small pairs created by the noise will screen the interaction between $a_1$ and $a_2$, leading to these anyons randomly diffusing and escaping correction. Since the time scale for the noise to create a pair of separation $r$ is exponentially long in $r$, this argument gives a memory lifetime scaling as $\tmem \sim e^{\xi_p}$. } 
	\end{figure*}

	When $Z=0$, our construction reduces to the message-passing decoder in ref.~\cite{lake2025fast}, which was proven to have a threshold for offline decoding. In this section, we show that it does {\it not} have a threshold for active decoding, meaning here that there exists a $p$-bounded error model for which $\lim_{L\ra\infty}\tmem < \infty$ for all $p>0$. We also show that the same is true for the field-based decoders of \cite{herold2015cellular,herold2017cellular}, even when the power-law interactions are set up instantaneously. In both cases, the absence of a threshold is due to the effects of transient errors, and occurs even when measurement outcomes are perfect.  
	
	For spatially local noise models, we fist give a general argument that the memory time scales (in $d$ dimensions) as 
	\be \label{pseudothresh} \tmem \sim \exp(b\min(L,c/p^{1/d})),\ee 
	where $b,c$ are $O(1)$ constants. This means that as a function of $L$, $\tmem$ increases exponentially until 
	\be L \gtrsim p^{-1/d},\ee 
	at which point $\tmem$ becomes constant. 
	With the analogy $p \sim e^\b$, this scaling mimics the partial-self correction present in several $3d$ quantum codes \cite{michnicki20143d,siva2017topological,williamson2023layer}. 
	Decoders for which $\tmem$ scales as in \eqref{pseudothresh} will be referred to as having a {\it pseudothreshold}, and in numerics can produce logical failure rates that can seem to quite convincingly mimic a threshold. 
	

	
	In the context of $1d$ classical error correction, an essentially identical argument to the one given in this section shows that the same pseudothreshold behavior occurs for the GKL automaton \cite{gacs1978one} and Toom's two-line voting rule, which were conjectured to have thresholds until proven otherwise for strongly biased noise in  \cite{park1997ergodicity}.\footnote{In fact, for strongly biased noise, the scaling is apparently the worse $\tmem \sim e^{\log^2(1/p)}$ \cite{park1997ergodicity}; our argument for $\tmem \sim e^{1/p^{1/d}}$ is an upper bound on $\tmem$ that we believe should be saturated for unbiased noise.} Indeed, the scaling $\tmem \sim e^{1/p}$ has been suggested (from numerics) to hold in the GKL model \cite{de1992gacs}, and our argument strongly suggests that a simpler proof---which does not rely on the noise being biased---is possible. It would be interesting to see whether or not the construction of ref.~\cite{paletta2025high} evades this problem.  
	
	\ss{Message passing with $Z=0$}\label{ss:nogo_message_passing}
	
	\begin{figure*}
		\centering 
		\makebox[\textwidth][s]{%
			\hspace{\fill}%
			\includegraphics[width=.4\textwidth]{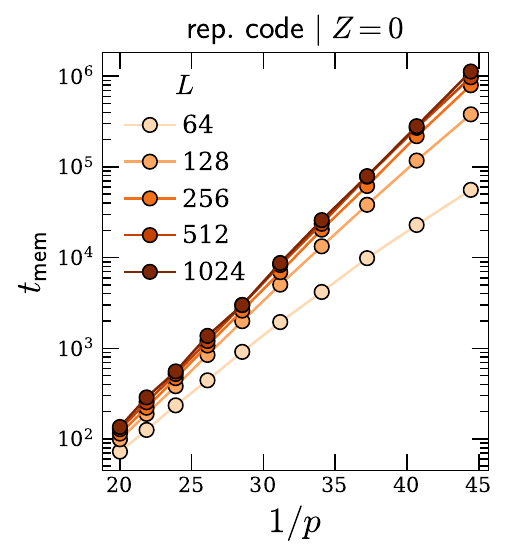}%
			\hspace{\fill}%
			\includegraphics[width=.395\textwidth]{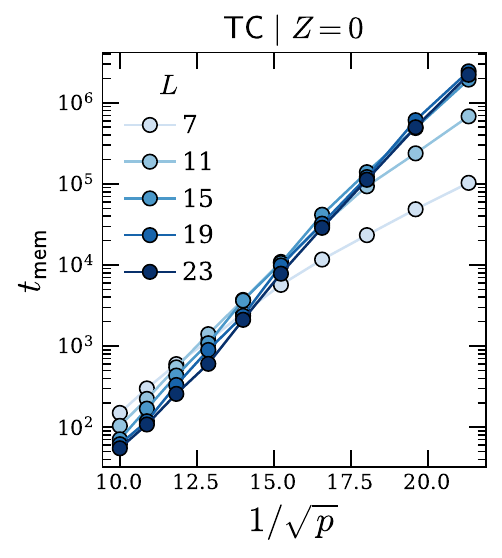}%
			\hspace{\fill}%
		}
		\caption{\label{fig:pseduo_scaling} Pseudothreshold behavior in $\tmem$ for message-passing decoders with $Z=0$, shown for the $1d$ repetition code (left) and the $2d$ toric code (right). } 
	\end{figure*}
	
	
	We begin by discussing the message-passing decoders with $Z=0$. The analysis of field-based decoders is similar, and will be discussed in the following subsection. The basic argument is illustrated in fig.~\ref{fig:screening_fig}. 
	
	To begin, let us first assume that $L = \infty$. At noise strength $p$, the typical distance $\xi_p$ between small anyon pairs created by the noise is, in $d$ dimensions, 
	\be \xi_p\sim p^{-1/d}.\ee 
	Now consider a pair of anyons $a_1,a_2$ separated by a distance $l$, which have been around long enough to have begun exchanging messages with one another. If $l \ll \xi_p$, $l$ is much more likely to decrease than increase at the next time step, since the anyons in the pair are more likely to be closer to one another than to any anyon created by the noise. On the other hand, if $l \gg \xi_p$, the closest anyon to e.g. $a_1$ is very unlikely to be $a_2$, and instead is likely to be one of the anyons created by the noise. Thus when $l \gg \xi_p$, the communication between $a_1$ and $a_2$ will be screened by the noise, and they will simply diffuse in random directions, without being corrected (see fig.~\ref{fig:screening_fig}). 
	
	As a result, $\tmem$ is simply set by the expected time it takes the noise to create a pair of size $\gtrsim \xi_p$ in a region of size $\sim \xi_p^d$.\footnote{Once such a pair has a large probability of having been created in any such region, offline decoding will fail with high probability. } To estimate this time, consider beginning with a small pair of separation $l=1$. 
	At a given time step, $l$ will increase only if at least two errors happen within a distance at most $l$ from the anyons in the pair. Since there are $\sim l^d$ places for such a noise event to happen,  if $l^d < 1/p$, then $l$ will increase only with probability at most $\sim (pl^d)^2$. The probability $p_{\sf grow}$ for a small error to grow to size $\xi_p$ is then 
	\be  \label{pgrow} p_{\sf grow} \sim \prod_{j=1}^{\xi_p} (j^d p)^2= (\xi_p!)^{2d} p^{2\xi_p} \sim e^{-2d/\xi_p},\ee 
	so that 
	\be \tmem \sim p_{\sf grow}\inv \sim e^{2d/p^{1/d}}.\ee 
	For finite $L$, this argument goes through as long as $L \geq \xi_p$. When $L<\xi_p$, the product in \eqref{pgrow} is instead cut off at $L$, and we get 
	\be \tmem \sim e^{-2d L(1+ \ln(\xi_p/L))}.\ee 
	In agreement with \eqref{pseudothresh}, $\tmem$ thus increases exponentially with $L$ until $L\sim\xi_p$, after which it becomes $L$-independent. 
	
	These predictions are verified numerically in fig.~\ref{fig:pseduo_scaling}, where we compute $\tmem$ for the $Z=0$ message passing decoders with i.i.d bit-flip noise of strength $p$ and perfect measurements, for both the $1d$ repetition code and the $2d$ toric code (code producing all numerical results in this paper is available at \cite{code}). In both cases we find $\tmem \sim e^{c/p^{1/d}}$ at the largest values of $L$, and our results are consistent with $\tmem$ plateauing at $L \sim \xi_p$. 
	
	Unfortunately, rigorously showing that \eqref{pseudothresh} holds for i.i.d noise appears to be slightly tedious. To rigorously establish the absence of a threshold, we therefore show the following: 
	\begin{theorem}[pseudothresholds for historyless message passing] \label{thm:nogo}
		When $Z=0$, there exists a $p$-bounded error model for which the message-passage decoder has 
		\be \label{z0tmembound} \lim_{L\ra\infty} \tmem < ae^{\frac{b}{\ln(1/(1-p))}} + c\ee 
		for positive $p$-independent constants $a,b,c$. 
	\end{theorem}
	The $p$-bounded error model referred to in the theorem statement is chosen in order to simplify the proof, and is rather adversarial. For the same reason, it is also chosen to be effectively one-dimensional, which produces the $1/\ln(1/(1-p)) \ra 1/p$ in the exponent. Since the physical argument and numerics given above indicates that these issues are just technicalities, we will not endeavor to work out the full details for the i.i.d case. 
	
	The idea of the proof is extremely simple, but there are several unilluminating details that need to be accounted for. We therefore give only a proof sketch here, deferring the rest to appendix~\ref{app:nogo}.

	{\it Proof sketch: } The basic idea is to choose a $p$-bounded noise model which always tries to ``expand'' large pairs of anyons. To simplify the analysis, it will be sufficient to restrict to noise models that perform bit-flip errors {\it only} along a given $1d$ strip of the system, and cause no measurement errors. 
	
	First let us define a noise model $\mce_{\sf s.p.}$ that tries to create and expand a single large anyon pair. Schematically, $\mce_{\sf s.p.}$ is defined as follows: 
	\begin{itemize}
		\item If there are no anyons in the system, $\mce_{\sf s.p.}$ creates a single pair of separation $r_0$ with probability $p^{r_0}$, where $r_0$ is some small constant. 
		\item Let $r_t$ be the separation of the largest anyon pair in the system at time $t$. If $r_t > 1$, the noise creates localized anyon pairs according to i.i.d $p$-bounded noise subjected to the following constraints: 1) each pair created by the noise is separated from all other anyons by at least two lattice sites, 2) pairs are created in a way which is reflection-symmetric about the midpoint of the pair of separation $r_t$, and 3) pairs are only created to the left of the left anyon in the pair, or to the right of the right anyon. 
	\end{itemize}
	\begin{figure}
		\centering 
		\includegraphics[width=.48\tw]{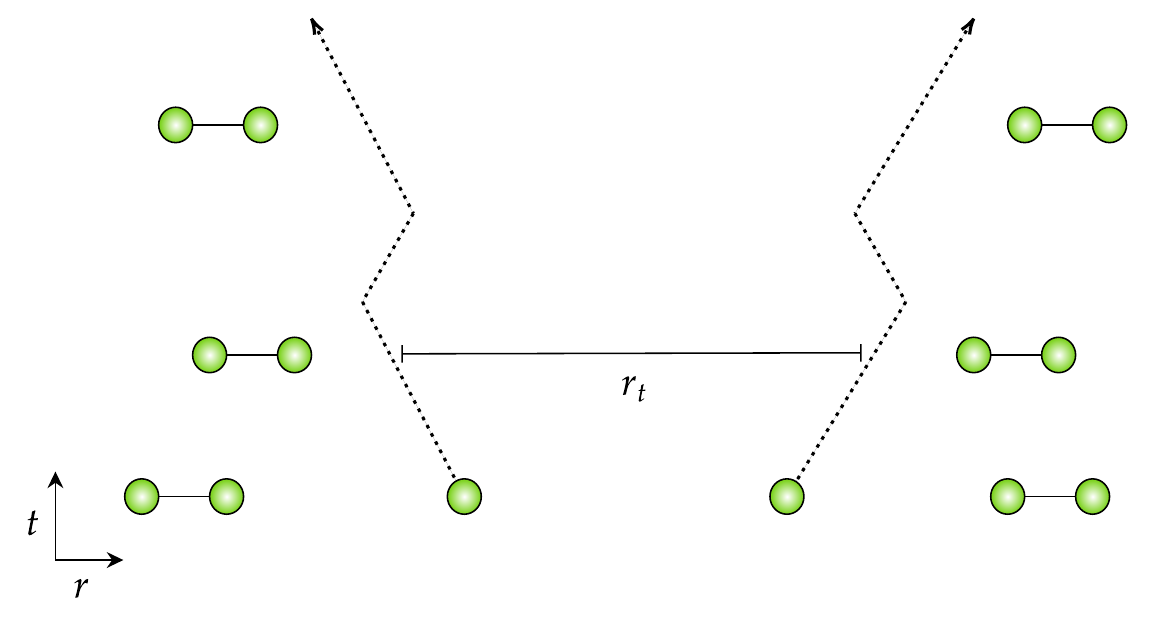} 
		\caption{\label{fig:simple_noise} A schematic of the error model used in the proof of theorem~\ref{thm:nogo}. The noise is restricted to create well-separated nearest-neighbor anyons at reflection-symmetric locations on the ``outside'' of any large pair.   }
	\end{figure}	
	For an illustration of these definitions, see fig.~\ref{fig:simple_noise}. The purpose of 1) is to ensure that each small anyon pair produced by the noise is eliminated at the next time step, and the purposes of 2) and 3) are to simplify the way that large anyon pair expands. It is easy to check that these additional constraints keep the noise $p$-bounded. 
	
	Under $Z=0$ message-passing dynamics with this noise model, the variable $r_t$ undergoes a particular type of random walk on $\zz^{\geq 0}$. In appendix~\ref{app:nogo}, we verify that the bias of this random walk is towards 0 when $r_t < r_* \sim 1/p$, and towards $\infty$ when $r_t > r_*$. Using standard results about hitting times in random walks, the expected time $\tau_l$ for $r_t$ to reach $l>r_*$ (viz., to unbind) is then  $\tau_l < ae^{b/\ln(1/(1-p))}+c$ for some $p$-independent constants $a,b,c$. 
	
	To show that $\tmem \sim \tau_{r_*}$, we may break the system up into boxes of width $\sim r_*$ (along the direction of the $1d$ strip along which the noise acts), and define a noise model $\mce$ that implements $\mce_{\sf s.p.}$ independently in each $r_*$-sized box (this model is thus spatially correlated only over scales of order $r_*$, which is independent of $L$). Since the expected time for an anyon pair to traverse each box is proportional to $\tau_{r_*}$, offline decoding fails with high probability after a constant multiple of this time scale. 		
	\qed 
	
	\ss{Field-based decoders} \label{ss:power_law_general}
	
	We now consider decoders that endow anyons with attractive power-law interactions, following \cite{herold2015cellular,herold2017cellular}. Engineering an arbitrary power law interaction in a local fashion cannot be done using the simple PDE-based approach of these references. While it can be done using message passing (as explained in \cite{lake2025fast}), getting the appropriate interaction requires increasing the number of message flavors, and slightly complicating the architecture. To simplify the discussion, we will therefore assume that the power-law interactions have been put in ``by hand'', without worrying about any local mechanism by which they could have been produced. We will also assume that the interactions are transmitted instantly fast, and will continue to assume that no measurement errors occur. Even in this setting, a threshold is unattainable. 
	
	Let $\{\bfr_i\}$ be the locations of anyons at a given time step. Define the force $\bfF_i$ experienced by the $i$th anyon as\footnote{Unfortunately, our definition of $\a$ differs by 1 from the $\a$ in \cite{herold2015cellular,herold2017cellular}.}
	\be \bfF_i = \sum_{j \neq i} \frac{\bfr_j - \bfr_i}{||\bfr_j - \bfr_i||^{\a+1}}.\ee 
	Following \cite{herold2015cellular,herold2017cellular}, we then take the decoder to, at each time step, move the $i$th anyon by one site in the lattice direction $\phi(\bfF_i)$, where 
	\be \phi : \rr^d \ra \{\pm \uva \, : \, a = 1,\dots, d\}\ee 
	is a function of our choosing; one natural choice is for $\phi(\bfu)$ to select out the unit lattice vector with largest dot product with $\bfu$. In the $v\ra\infty$ limit, the message passing decoder with $Z=0$ becomes equal to the power-law decoder with this choice of $\phi$ when $\a \ra \infty$. 
	
	Consider a well-separated anyon pair with relative separation $\bfr$. A threshold can be achieved only if the anyons are attracted to one another by a force strong enough to prevent them from diffusing to larger separations. We will show that if $\a$ is too large, the attractive force is too weak, and gets screened by the signals sent from nearby noise events (as is the case for $Z=0$ message-passing). On the other hand, if $\a$ is too small, the signals from very distant noise events dominate. In $d>1$ dimensions, there is no intermediate range of $\a$ where at least one of these problems is not present. 
	
	To see this from a quick physical argument, write $d\bfr = d\phi(\bfF)$, where $\bfF$ is thought of as having a random component (from the noise) and a deterministic attractive component (coming from the two anyons in the pair), 
	\be \bfF= -\frac{\bfr}{r^{\a + 1}} + \bfxi,\ee 
	where $\bfxi$ is the force induced by the random arrangement of small noise-induced anyon pairs. The Fokker-Planck equation for the probability density $\r(\bfr)$ to find the anyon pair at separation $\bfr$ is 
	\be \p_t \r = \D_a ( \lan \phi^a(\bfF)\ran \r +\frac12 \D_b( \lan \phi^a(\bfF) \phi^b(\bfF)\ran  \r),\ee
	giving 
	\be \p_t\r = \D_a \( v^a \r + \frac12 D \D^a \r \)\ee 
	where the diffusion constant $D$ is determined by the noise correlation $\lan  \xi^a \xi^b\ran$ (which we have taken to be $\propto \d^{ab}$), and the drift velocity---coming from the attractive confining interaction between the anyons in the pair---is $\bfv = -\phi(\bfr / r^{\a + 1})$. 
	
	For spatially homogeneous noise, we may take an angular average and write 
	\be \p_t\r_R = \p_r \( \wt v_r \r_R + \frac D2\p_r \r_R\),\ee 
	where $\r_R = r^{d-1} \r$ is the radial PDF and the radial velocity $\wt v$ is (letting $\phi(\bfu) = \hat\bfu$ for concreteness)
	\be \wt v_r =- \frac1{r^\a} + \frac{(d-1)D}{2r}.\ee 
	To beat the outward bias caused by the entropic effect of being in $d>1$ dimensions, we clearly need $\a \leq 1$, as if $\a>1$ the anyon pair will simply diffuse as long as $r > \xi_p$, and suffer the same screening problem as our $Z=0$ message-passing architecture. However, if we take $\a \leq 1$, the interactions are so strong that distant noise events override the signals that try to correct the anyon pair. Indeed, for noise which creates localized anyon pairs according to a Poisson point process of density $\sim p$ (which we assume to be immediately corrected at the next time step), the variance $\s^2_F$ in the signal caused by the noise is 
	\be  \s^2_F = \lan \phi^a(\bfxi)\phi^b(\bfxi) \ran \sim \sum_{i,j} \frac{\lan r_i^a r_i^b\ran }{r_i^{\a+1} r_i^{\a+1}} \sim p \int d^2r \frac1{r^{2\a }} ,\ee 
	which diverges as $L\ra\infty$ when $\a \leq 1$. Since the fluctuations in $\bfF$ are divergent when $d>1$ and $\a\leq 1$ (giving $D \ra\infty$ and hence $\wt v_r > 0$), the signal any anyon receives is random, and the anyons in the pair again just execute uncorrectable diffusive motion. 			
	In $1d$, this argument leaves open the possibility of having a memory when $1/2 < \a < 1$ (values for which $\xi$ has finite variance and $\a<1$). For i.i.d noise, it is hard to say from numerics (see below) whether or not this is case, although we regard this as being unlikely. 
	
	As before, we can rigorously demonstrate the absence of a threshold in a particular $p$-bounded noise model (in all dimensions, and for all values of $\a$), following essentially the same analysis as in the message-passing case. Since the details are unilluminating, they are deferred in their entirety to appendix~\ref{app:power_law_nogo}.

	\begin{figure*}
		\hfill \includegraphics[width=.32\tw]{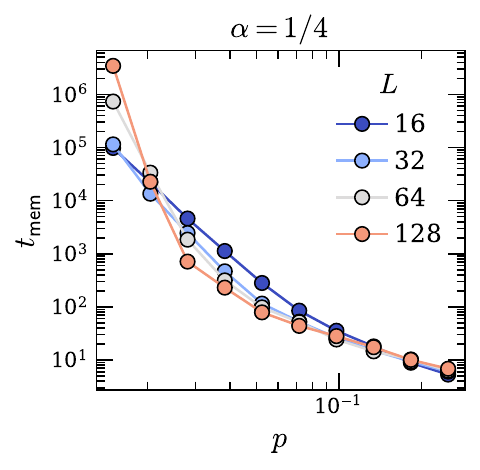} \hfill 
		\includegraphics[width=.32\tw]{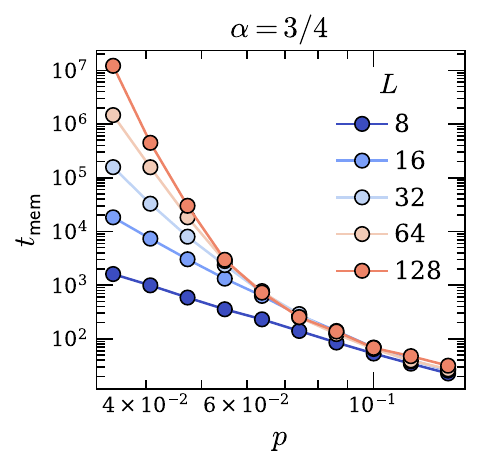} \hfill 
		\includegraphics[width=.32\tw]{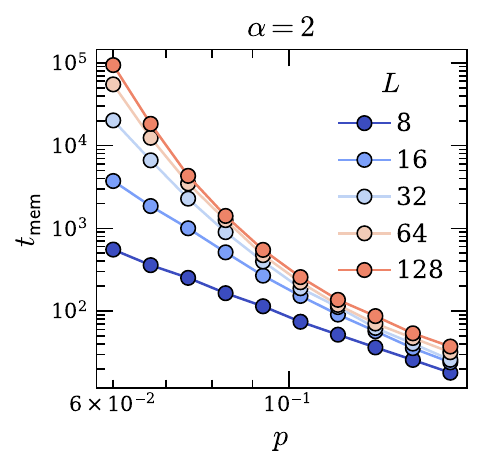} \hfill 
		\caption{\label{fig:no_thresh_power_laws} Memory lifetimes for $1d$ power law decoders under i.i.d noise of strength $p$ and perfect measurements. {\it Left:} $\a = 1/4$, where the noise variance $\s_F^2$ diverges as $\sqrt{L}$ when $L\ra\infty$. This results in $\trel$ decreasing with $L$ over an intermediate range of $p$. {\it Center:} $\a = 3/4$, where $\s_F^2$ is finite and the interactions are strong enough to localize the $1d$ random walk considered in the main text. The data may be consistent with a true threshold, but a pseudothreshold scaling appears more likely (in either case, a threshold is not present in $d>1$, or under the $p$-bounded error model of appendix~\ref{app:power_law_nogo}). {\it Right:} $\a = 2$, which behaves similarly to the local message-passing decoder. }
	\end{figure*}
	
	We now investigate the memory lifetime in Monte Carlo simulations. 
	We focus on the $1d$ repetition code, where the problem of diverging noise signals at small $\a$ is least severe. The results are shown in  fig.~\ref{fig:no_thresh_power_laws}. 
	In the left panel, we set $\a=1/4$. In this case the variance diverges as $\lim_{L\ra\infty}\s_F^2 \sim \sqrt L$, which manifests as a {\it decrease} in $\tmem$ with $L$ in an extended intermediate range of $p$. 
	The center panel shows $\a=3/4$; here the expected force caused by signals coming from one side of a given anyon is infinite but $\lim_{L\ra\infty}\s_F^2<\infty$, and the interaction is strong enough to bind the random walk considered above (but only because $d=1$). Our arguments are consequently not able to rule out a threshold in this case (although no threshold is present for the $p$-bounded noise model used in appendix~\ref{app:nogo}). 
	Finally, the right panel shows $\a =2$. Here $\lim_{L\ra\infty}\s_F^2 < \infty$, and we expect $\tmem$ to behave similarly to the message-passing decoder with $Z=0$; the right panel of fig.~\ref{fig:no_thresh_power_laws} is indeed consistent with this expectation. 

	We close this section with a brief remark on another field-based decoding scheme which lacks a threshold. In some situations, it may be physically natural to assume that the anyons carry a conserved $U(1)$ ``electric'' charge which the noise preserves strongly, viz. which has an associated symmetry generator that commutes with every Kraus operator applied by the noise. Putting aside the non-generic requirement of the symmetry preservation being strong, this scenario could arise naturally in quantum Hall states (such as e.g. the $\nu=1/3$ Laughlin state), or less naturally in certain lattice models.\footnote{E.g. in a modified toric code with additional spins $Z_i, Z_\square$ on the vertices and plaquettes, and with strong symmetries generated by both $Q_e = \sum_i A_i Z_i$ and $Q_m =  \sum_\square B_\square Z_\square$, with $A_i, B_\square$ the conventional vertex and plaquette stabilizers.} In this scenario, we could take the interactions between anyons to reflect their charges by letting 
	\be \bfF_{i} = q_i \sum_{j\neq i} q_j\frac{\bfr_j - \bfr_i}{|\bfr_i - \bfr_j|^{\a +1}},\ee 
	where $q_i\in \zz$ is the charge of the $i$th anyon. 
	
	Naively, this might let us circumvent the general no-go argument in sec.~\ref{ss:power_law_general} if $\a$ is chosen appropriately: since anyons are always created in charge-neutral pairs, the strength of the contribution to $\bfF_i$ from noise-created pairs falls off as $1/r^{\a +1}$ instead of $1/r^\a$ (since it is a dipolar force), while the force from the anyon that anyon $i$ is paired with falls off as $1/r^\a$. This reasoning is however overly optimistic, and in appendix~\ref{app:charged} we show that a threshold is also not present in this case. 
	
	\section{Decoding performance}\label{sec:results}
	
	In this section we 1) prove that our message-passing decoder has a threshold for active decoding as long as $Z = \O(\polylog(L))$, 2) show how its operation can be made asynchronous, 3) identify the threshold error rate in numerics, 4) provide a preliminary investigation of the phase transition that occurs in the $1d$ model, and 5) numerically study our decoder's ability to perform fault-tolerant initialization. 
	
	\ss{Threshold theorems} \label{ss:thresholds}
	
	In this section, we prove the existence of a threshold when $Z$ is large enough. To simplify the proof, we will slightly modify the feedback rules used by the decoder. First, we will add a bias in the defect motion towards $+\uvx$: if 
	\be (s,a) = {\rm argmin}_{s', a'}\{m^{s'a'}_\bfx\} \not\in \{(1,x), (-1,x)\},\ee 
	we move the defect first along $s\uva$, and then along $+\uvx$ (if the defect did not annihilate after the first movement step). The same change is made both in the bulk and on the back wall. The reason for doing this is purely technical, and was explained in \cite{lake2025fast}. Secondly, we increase the speed of the flow along $+\uvz$ by 1, replacing $\bfx-\uvz$ on the RHS of \eqref{rg_cycle} by $\d_{x^{d+1},2}(\bfx-\uvz) + (1-\d_{x^{d+1},2})(\bfx-2\uvz) $. This ensures that all defects move to larger $z$ coordinates at all time steps, even if the nearest message they recieve is from the $-\uvz$ direction. 
	
	Using this modified message-passing dynamics, we show 
	\begin{theorem}[existence of a threshold]\label{thm:thresh}
		There are $O(1)$ positive constants $0<\z,\b,p_c<1,a$ such that as long as $v>2$ and 
		\be Z  \geq  Z_{\sf max} = a \(\frac{\log(L)}{\log(p_c/p)} \)^{1/\z},\ee 
		then for any $p$-bounded error model with $p<p_c$, the memory lifetime of the modified message-passing decoder satisfies 
		\be \label{tmemscaling} \tmem = (p_c/p)^{\O(L^\b)}.\ee 
	\end{theorem} 
	We will prove this result assuming that the offline decoder $\mcd_{\sf off}$ used in definition~\ref{def:tmem} is simply the online message-passing decoder run noiselessly until all defects are eliminated, although other choices are certainly possible. Additionally, while this result applies for both synchronous and asynchronous updates, for simplicity we will specialize to the former case for now. The generalization to asynchronous updates is provided in sec.~\ref{ss:lind}. 
	
	We begin by showing under what conditions a spacetime cluster of errors can reach the back wall. 
	\begin{lemma}\label{lemma:cluster_erosion}
		Consider a $k$-cluster $\mcc^{(k)}$ of spacetime errors. There is an $O(1)$ constant $c$ such that if 
		\be k < k_Z = \log_n(Z/c),\ee 
		then all defects in $\mcc^{(k)}$ are annihilated against other defects in $\mcc^{(k)}$ before reaching the back wall, provided $v>2$. 
	\end{lemma}
	As a consequence, noise containing only clusters of level $k < k_Z$ can never cause a logical failure. The proof is quite similar to the proof of linear erosion under message-passing dynamics given in \cite{lake2025fast}, and we defer it to appendix~\ref{app:details}. 
	
	In light of this result, the back wall can be regarded as undergoing $d$-dimensional message passing decoding in the presence of noise that contains no clusters of level less than $k_Z$. Roughly speaking, this then corresponds to a renormalized noise model with error rate 
	\be p_Z = p^{2^{k_Z}} = p^{(Z/c)^\z},\qq \z = \log_n(2).\ee 
	Since a logical error can only occur when the back wall has accumulated a nonzero number of errors, this gives the easy lower bound 
	\be \tmem \geq a\frac{p_Z}{L^d} = ae^{\ln(1/p) (Z/c)^\z - d \ln L},\ee 
	where $a$ is an unimportant $O(1)$ constant. For any $\l>1$, this then produces a threshold with 
	\be \tmem = (p_c/p)^{\O((\log L)^\l)},\ee 
	provided 
	\be Z = \O((\log L)^{\l/\z}).\ee  
	
	This simple bound already produces a superpolynomial memory lifetime for $Z = \O(\polylog(L))$, but it is much too conservative, since a single $k_Z$-cluster that survives to hit the back wall cannot by itself cause an error (provided $Z \ll L$), and instead will itself be corrected by the message passing dynamics on the back wall with high probability. 
	To get a better bound, we may use essentially the same argument as in sec.~\ref{ss:nogo_message_passing}, except with $p$ replaced by $p_Z$. In this argument, one uses that the typical spacing between $k_Z$-clusters that survive to imprint themselves on the back wall is $\sim p_Z^{-1/d}$, and a $k$-cluster must then have level 
	\be k > \log_n(cp_Z^{-1/d})\ee 
	in order to not be ballistically shrunk by the message-passing dynamics. Creating such a cluster requires a noise event of weight at least  $\sim2^{\log_n(p_Z^{-1/d})}$, which becomes equal to the weight of a $\lfloor\log_n(L)\rfloor$-cluster when $Z = \polylog(L)$, thereby giving the scaling in \eqref{tmemscaling}.  
	
	We now make this argument more rigorous. To do so, we need some definitions. We will first construct a graph defined from the pattern of defects arranged on the back wall. Let $\{\mcc_\a\}$ denote the set of spacetime noise clusters, and let $\bw(\mcc_\a,t)$ denote the set of defects on the back wall at time $t$ which originated from cluster $\mcc_\a$ (an association of a back-wall defect to a unique noise cluster is well-defined because the message-passing dynamics on the back wall never creates new defects). Form a graph $\scg(t)$ as follows: 
	\begin{itemize}
		\item The nodes of $\scg(t)$ are those clusters such that $\bw(\mcc_\a,t) \neq \emp$. 
		\item If both $\mcc_\a$ and $\mcc_\b$ are nodes in $\scg(t)$, an edge is drawn between these nodes if a defect in $\mcc_\a$ has fused with one in $\mcc_\b$ by time $t$. 
	\end{itemize}
	
	\begin{definition}[clumps]
		Let $\scg_\l(t)$ be a particular connected component of $\scg(t)$. A clump $\scc_\l$ at time $t$ is the set of back-wall defects contained in the clusters constituting this component at this time: 
		\be \scc_\l = \bigcup_{\mcc_\a \in \scg_\l(t)} \bw(\mcc_\a,t).\ee 
		We say that $\scc_\l$ has size $r$ if $\scc_\l$ is contained in an $\infty$-norm ball of width $r$. 
	\end{definition}
	
	The occurrence of a logical error mandates that a clump of size $\ct(L)$ be created by the noise and a subsequent noiseless evolution (the latter coming from the $\mcd_{\sf off}$ in \eqref{tmemdef}). To show how unlikely such an event is, we need the following lemma: 
	\begin{lemma}\label{lemma:kmax}
		Consider an event where $m$ noise clusters $\{\mcc_\a^{(k_\a)}\}$ of levels $k_\a$ create a clump of size $\ell$. Then there exist $O(1)$ positive constants $\g,g$ such that 
		\be \label{kmax} \max_\a k_\a \geq \max \(\log_n(\ell/(gm^\g)),k_Z\).\ee 
	\end{lemma}
	This means that, when $m$ is small, creating a clump of size $\ell$ requires that at least one of the clusters involved in its assemblage be not too much smaller than $\ell$. The proof is given in appendix~\ref{app:details}. 
	
	We now use this result to complete the proof of theorem~\ref{thm:thresh}. In what follows, otherwise-undefined roman letters will represent unimportant $O(1)$ positive constants ($d$ will continue to denote the spatial dimension). 
	\begin{proof} 
		As mentioned above, a logical failure can occur only when a clump of linear size $cL$ forms on the back wall. Consider a scenario where such a clump is formed from $m$ noise clusters at levels $k\geq k_Z$, where all the noise clusters that reach the back wall between the formation of the clump and the previous time where the back wall contained no clumps are included. From lemma \ref{lemma:kmax}, at least one of these clusters be of level $\kmax \geq \log_n(o L / m^\g)$, and all the rest trivially must be at least of level $k_Z$. The total weight $w$ of the noise in this set of clusters is thus at least 
		\be w \geq  (m-1)2^{k_Z} + \(\frac{oL}{m^\g}\)^\z\ee 
		where recall that $\z = \log_n(2)$. 
		
		We now need to determine the entropy of such a set of clusters. Each of the $m$ clusters must be located within a temporal distance of $fL$ of another such cluster: if this were not true, linear erosion implies that the clump would have been eliminated by back-wall message passing. Since there are at most $\log_nL$ different levels of clusters to choose from, the total number of ways $N_m$ of getting $m$ such clusters is very loosely upper bounded by 
		\be N_m \leq (fL)^{m(d+1)} (\log_n(L))^m.\ee 
		Therefore, the probability $p_{\sf err}(T)$ of an error happening in time $T$ is 
		\bea p_{\sf err}(T) & \leq T \sum_{m=1}^\infty \exp\big(m(d+1)\ln(fL) + m\log_n(L) \\ & \qq - \ln(p_c/p) \((m-1)Z^\z+ (oLm^{-\g})^\z \)\big). \eea 
		Assuming $Z^\z > q\ln(L)/\ln(p_c/p)$, the saddle point is at 
		\be m_* = \ct((L/Z)^{\z / (1+\g \z)}),\ee 
		and so
		\be p_{\sf err}(T) \leq T p^{\O(L^{\z / (1+\g \z)} Z^{\g \z^2 / (1+\g \z) })}, \ee 
		which gives the desired result. 
	\end{proof}

	\ms 
	\begin{corollary}[generalization to arbitrary Abelian anyon models] \label{cor:generalization}
		Consider using the same message-passing architecture to perform error correction on an arbitrary topological code containing only Abelian anyons: all defects emit messages (regardless of their associated topological charge), and defect motion is additionally independent of topological charge. In this setting, the results of theorem~\ref{thm:thresh} hold. 
	\end{corollary}
	
	\begin{proof}
		This is a direct consequence of clustering and the fact that since the anyons are Abelian, the correction operators applied when defects at different $z$ coordinates are moved commute up to a phase (see also the discussion in \cite{lake2025fast}). 
	\end{proof}
	
	\ss{Desynchronization} \label{ss:lind} 
	
	
	As discussed in the introduction, synchronicity is a nonlocal resource, which a self-organized error-correction process must not rely on. In what follows we discuss various desynchronization schemes to remove this dependence. 
	
	\sss{Asynchronous Non-Markovian updates} \label{sss:asynch_nonmark}
	
	In order to rigorously prove a threshold in a model with asynchronous updates, it is most convenient to work with a {\it non-Markovian} continuous-time update scheme. In this scheme, each processor $\sfP_\bfr$ updates after waiting times $\d_\bfr$ that are i.i.d random variables distributed as 
	\be \label{nonmarkupds} \d_\bfr \sim [1-\ep,1+\ep],\ee 
	where $\ep < 1$ is a positive constant. This update schedule is fully asynchronous, but mandates that a small counter be present at each site to keep track of the time elapsed since the last update.\footnote{Since we are already assuming access to reliable classical bits on each site, we will not be concerned with this small additional requirement.} 
	
	To analyze this update scheme, we need to extend our noise models to the present continuous-time setting. We will do this by marking out the points in spacetime where each processor updates,  and working with an effectively discrete-time $p$-bounded noise model defined on this set of spacetime points (for the techniques used to get a noise model like this from a more realistic continuous-time process see e.g. \cite{gottesman2024surviving}). 
	
	The point of using this update scheme is that in a time interval of duration $T$, the minimal and maximal number of updates any given processor makes are bounded by quantities proportional to $T$. This makes our threshold proof extend quite directly to the asynchronous case.
	
	To simplify the proofs, we will need to first slightly embellish the automaton rules to prevent defects from being moved too fast. This is done by enlarging the classical variable space to include ``cooldown'' variables $c_\bfx \in \{0,1\}$, and adding the following additional automaton rules: 
	\begin{enumerate}
		\item If $s_\bfx=-1$, feedback may only be applied to move the defect at $\bfx$ to $\bfx'$ if $c_\bfx = 0$. 
		\item If an update applies feedback to move a defect to $\bfx'$, then $c_{\bfx'} \mt 1$. 
		\item $c_\bfx$ is reset to $0$ each time $\bfx$ updates. 
	\end{enumerate}
	This prevents a conspiratorial arrangement of update times (arranged in a ``staircase'' with small step heights) from moving defects arbitrarily quickly. Note that we do not include cooldown fields for the $m^{\pm a}_\bfx$, so the messages in principle may be transmitted arbitrarily quickly (but not arbitrarily slowly). 
	We take the vertical drift along the $+\uvz$ direction to be implemented on messages and defects each time they move (again taking each drift step to be by two lattice sites). 
	
	There are several ways of implementing a message velocity $v>1$ in this setting. We will choose to simply take a message propagation step at $\bfx$ to result in the messages $m^{\pm a}_\bfx$ propagating to all sites that they would have propagated to under $v$ steps of synchronous message passing. This choice ensures that the longest time $t_{\sf max}(l)$ a message can take to propagate a distance $l$ is 
	\be\label{vmin} t_{\sf max}(l) = l\frac{1+\ep}{v}.\ee 
	We may thus think of the messages as having a velocity distributed between $v/(1+\ep)$ and $\infty$, depending on the (stochastic) choice of update schedule. 
	
	\begin{proposition}[threshold with asynchronous updates]
		For the message-passing decoder with cooldown variables and updates applied according to \eqref{nonmarkupds}, the threshold results of theorem~\ref{thm:thresh} hold, provided $v > 2(1+\ep)/(1-\ep)$. 
	\end{proposition}
	\begin{proof}
		In a time interval $T$, any given processor experiences $n_{{\sf up}}(T)$ updates, where $n_{\sf up}(T)$ satisfies 
		\be \label{nupbounds} \frac{T}{1+\ep} \leq n_{\sf up}(T) \leq \frac{T}{1-\ep}.\ee 
		This ensures that the correction of a cluster of defects is (at worst) only slowed down by a constant factor, which does not compromise either lemma~\ref{lemma:cluster_erosion} or lemma~\ref{lemma:kmax}, and hence does not change the existence of a threshold (although it will in general reduce the value of $p_c$). 
		
		In slightly more detail, consider showing that  lemma~\ref{lemma:cluster_erosion} continues to hold with this update scheme. Let $\De a(t)$ denote the amount by which a particular defect moves along the $\pm\uva$ direction during a time $t$.
		Because of the cooldown fields,
		\be \label{deltaa} \De a(t) \leq \frac{(2\d_{a,z} + 1)t}{1-\ep} .\ee 
		Consider a 0-cluster $\mcc^{(0)}$. If the first defect in $\mcc^{(0)}$ enters the system at time $t_0$, all of the defects it contains will be at $z>0$ by time $t_0 + w(1+\ep)$, and by this time they will be contained within a ball of linear size not more than $3w(1+\ep)/(1-\ep)$. Because of \eqref{vmin}, by a time $cw$ all defects in $\mcc^{(0)}$ will be at $z>0$, and each will have another defect in $\mcc^{(0)}$ in its past light front. A straightforward generalization of the erosion property proved in \cite{lake2025fast} then shows that defects in $\mcc^{(0)}$ annihilate among themselves before reaching defects in any other cluster. The synchronous version of this result requires that the maximum relative velocity between defects be less than the message velocity, which in the present setting translates into the inequality (using \eqref{vmin} and \eqref{deltaa})
		\be \frac2{1-\ep} < \frac{v}{1+\ep}.\ee 		
		The rest of the argument then proceeds as in the synchronous case, and the generalization of lemma~\ref{lemma:kmax} and the subsequent arguments in the proof of theorem~\ref{thm:thresh} are similar. 
		
		This shows that lemma~\ref{lemma:cluster_erosion} holds in the present setting. The remainder of the proof is basically identical to that of theorem~\ref{thm:thresh}; the only things that change are the unimportant $O(1)$ constants appearing at various places. 
	\end{proof}

	\sss{Poissonian desynchronization} 
	
	From a physical perspective, the most natural asynchronous scheme is a Markovian one whereby each processor updates according to independent Poisson processes. This allows the dynamics to be straightforwardly implementable by a local Lindbladian, but makes proofs of thresholds rather cumbersome due to the ballistic correction of isolated clusters of defects now being only an event which holds with high probability. Due to its simplicity, we will use this scheme when numerically studying the performance of our decoders in Sec.~\ref{ss:threshold_numerics}.
	
	\sss{Marching soldiers}
	
	A simple way of desynchronizing synchronous computations is the {\it marching soldiers} desynchronization scheme (see e.g. \cite{gacs2001reliable,cook2008self}), which was used in ref.~\cite{lake2025fast} to convert synchronous offline decoders into asynchronous Lindbladian ones. 
	As a brief review, the marching soldiers scheme works by having each processor $\sfP_\bfr$ propose updates---which may be rejected---according to independent Poisson processes of a fixed common rate. Let $\tsim(\bfr,t)$ denote the number of updates of the synchronous computation that processor $\sfP_\bfr$ has simulated. In the absence of noise occuring during the computation, any desired synchronous computation can be simulated in the asynchronous system by requiring that $\sfP_\bfr$ accept a proposed update only if $\tsim(\bfr,t) \leq \tsim(\bfr',t)$ for all sites $\bfr'$ neighboring $\bfr$. $\sfP_\bfr$ thus advances only if doing so would not ``leave behind'' any of its neighboring processors.

	Let us refer to the set $\{ \tsim(\bfr,t)\}$ at a fixed $t$ as the {\it simulation surface} of the computation at time $t$. 
	One may show that with high probability, any point $\bfr$ on the simulation surface has $a t \leq \tsim(\bfr,t) \leq b t$ for $O(1)$ constants $a,b$ \cite{berman1988investigations,cook2008self,lake2025fast}, so that the rate at which the asynchronous computation progresses differs from that of the synchronous computation by only a constant factor. However, this fact is {\it not} enough to guarantee that the asynchronous computation is reliable in the presence of noise. This is because processors located in regions of the simulation surface where long slopes are present will go long periods of time without updating, and will thus accumulate errors at a higher rate than they would in the synchronous computation. 
	
	To understand the prevalence of these ``delay regions'', it is illuminating to note that the manner by which the simulation surface grows is identical to the way surfaces grow in the restricted solid-on-solid model of surface growth \cite{kim1989growth}. A renormalization group analysis \cite{kim1989growth} shows that the surface is always algebraically rough,\footnote{The {\it equilibrium} restricted solid-on-solid model has a smooth phase. Out of equilibrium (the setting relevant to marching soldiers) there is nothing to prevent the KPZ term from being generated under RG; once generated, it precludes a smooth phase.} here meaning that the connected correlation functions of the simulation time behave as $\lan \tsim(\bfr,t) \tsim(\bfr',t)\ran_c \sim |\bfr-\bfr'|^{2\l}$, where $\l > 0$ is a (dimension-dependent) constant. This roughness complicates a rigorous study, and is perhaps why ref.~\cite{cook2008self} was unable to resolve the question of whether or not the  marching soldiers scheme can be used for  synchronization in the setting of active error correction. 
	
	\subsection{Gerrymandering}\label{ss:gerry}
	
	In this section we show that a price must be paid for doing decoding locally, in that the sub-threshold logical error rate is not suppressed with $L$ as quickly as it is under optimal (nonlocal) decoding. This property in fact holds for all local decoders in the literature the author is aware of, including classical memories like Toom's rule, and will be discussed in more detail in future work.
	
	We begin with some definitions: 
	\ms 
	\begin{definition}[distance of a decoder]
		Consider a decoder $\mcd$ for a code $\mcc$. The distance $d_\mcd$ of the decoder is the smallest weight of a physical error that causes a logical error.\footnote{For active decoding, by ``causing a logical error'' we mean that a logical error is produced after running an appropriate offline decoder (which in this discussion will always be taken to be the active decoder in the absence of noise).}
	\end{definition}
	
	For optimal decoding we have $d_\mcd = \ct(d_\mcc)$, where the code distance $d_\mcc$ is the smallest weight of a nontrivial logical operator (and $\ct$ denotes scaling with respect to $L$). In general though, $d_\mcd$ may be asymptotically less than $d_\mcc$. When this occurs, we will refer to $\mcd$ as being ``Gerrymandered'', since a small minority of errors can lead to a logical failure:  
	\ms 
	\begin{definition}[Gerrymandering] \label{def:gerry}
		Consider a code $\mcc$ with code distance $d_\mcc$ and a decoder with distance $d_\mcd$. $\mcd$ is said to be {\it Gerrymandered} if 
		\be d_\mcd = o(d_\mcc).\ee 
	\end{definition}
	Note that in our definition we require $d_\mcd$ to be asymptotically (in the $L\ra\infty$ limit) less than $d_\mcc$; just being smaller by a constant amount does not count. 
	
	In the following, we show that not only are the message passing decoders Gerrymandered, but that under i.i.d noise, the entropy of the Gerrymandered events is large enough to affect the asymptotic scaling of $\tmem$ as $L\ra\infty$: 
	\ms 
	\begin{theorem}[message passing decoders are Gerrymandered on average]\label{thm:typ_gerry}
		Consider the limit where $L$ is taken to $\infty$ with $p < p_c$ held fixed. Then under i.i.d noise, there exists a constant $p_* \leq p_c$ such that when $p<p_*$, 
		\be \label{ploggerry} \tmem = (p_c/p)^{o(L^{\log_6(5) + \ep})} \ee 
		for any $\ep > 0$.
	\end{theorem} 
	The i.i.d assumption on the noise simplifies the proof but is likely not necessary (and in any case, most correlated error models one would reasonably consider would give larger values of $\plog$). Note also that in the theorem statement we have introduced the Gerrymandering scale $p_*$ as being separate from $p_c$, and our analysis does not rule out the possibility that $p_*<p_c$. The proof is given in appendix~\ref{app:details}. 
	
	\begin{figure*}
		\makebox[\textwidth][s]{%
			\hspace{\fill}%
			\includegraphics[width=.44\textwidth]{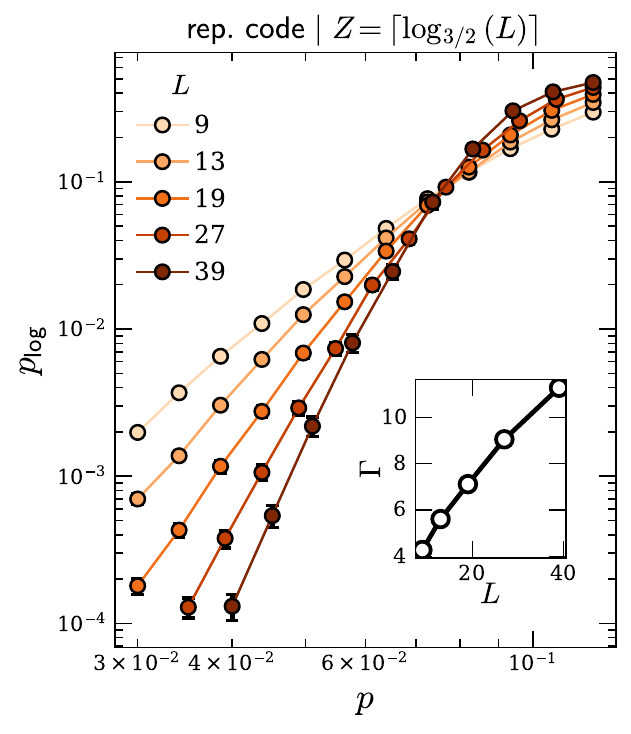}%
			\hspace{\fill}%
			\includegraphics[width=.45\textwidth]{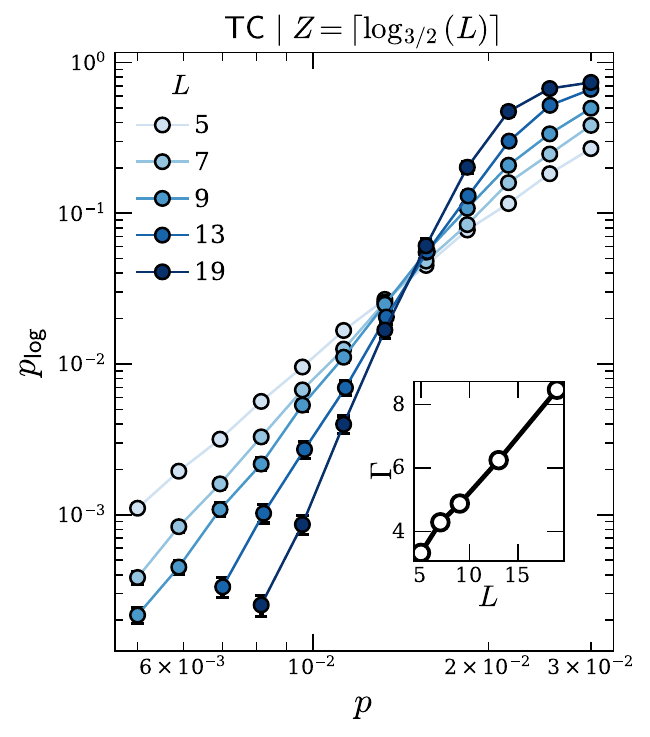}%
			\hspace{\fill}%
		}
		
		\caption{\label{fig:plog_vs_p} Scaling of the logical error rate with synchronous updates and i.i.d Pauli and measurement noise of strength $p$, shown for the $1d$ repetition code (left) and the $2d$ toric code (right). In both plots, we take $Z = \lceil \log_{3/2}(L) \rceil$. The values of $L$ are chosen to be the smallest odd integers which yield a strictly increasing sequence of $Z$ values. The insets show the power $\G$ appearing in a fit of the sub-threshold logical error rate to the form $\plog = (p/p_c)^\G$. Error bars are drawn at two standard deviations.  
		} 
	\end{figure*}

	\ss{Numerics}\label{ss:threshold_numerics}

	We now numerically study the performance of the message-passing decoders\footnote{In numerics we will use the architecture described in sec.~\ref{ss:setup}, instead of the modified (more analytically friendly) construction of sec.~\ref{ss:thresholds}.} for the $1d$ repetition code and the $2d$ toric code using Monte Carlo simulations (code generating the plots to follow is available at~\cite{code}). All of our simulations will assume i.i.d Pauli noise and measurement errors of strength $p$, and we will fix the message velocity at $v=3$ throughout. When computing $\plog$, we will decode by taking $\mcd_\off$ to be the noiseless message-passing decoder in $2d$, and a simple global majority vote in $1d$. 
	
	The scaling of $\plog$ for synchronous updates is shown in fig.~\ref{fig:plog_vs_p}. We set $Z = \lceil \log_{3/2}(L)\rceil$, the base of the logarithm being chosen so that an array of different $Z$ values may be reached without going to very large system sizes. This scaling is strictly speaking not fast enough with $L$ to guarantee a threshold from the results of theorem~\ref{thm:thresh}, but this is likely due to the rather loose bounds used during the proof (and in any case, a difference would likely not be visible at the present system sizes). 
	
	For the $1d$ repetition code a threshold is observed at 
	\be p_{c,1d}  \approx 7.5\%,\ee 
	which as expected is very close to the threshold the message-passing decoder achieves for {\it offline} decoding in the $2d$ toric code \cite{lake2025fast}, and only a few percent below the optimal value of $\approx 11\%$ \cite{Dennis_2002}. Below threshold, we fit the logical error rate to a function of the form $\plog = (p/p_c)^{\G}$ for an $L$-dependent number $\G$, whose value for different system sizes is shown in the inset. We observe a nearly-linear dependence of $\G$ on $L$, although the results of sec.~\ref{ss:gerry} mandate that $\G = o(L)$ at small enough $p$ (regardless of $Z$). 
	
	For the $2d$ toric code we find a threshold at 
	\be p_{c,2d} \approx 1.5\%,\ee 
	which is again very close to the offline decoding threshold of the message-passing decoder in $3d$, and respectably close (for a local decoder) to the optimal value of $\approx 2.9\%$ \cite{wang2003confinement,stephens2014fault}. The scaling of $\G$ with $L$ is similar as in $1d$.

	\begin{table*}[ht]
		\centering
		\renewcommand{\arraystretch}{1.5}
		\begin{tabular}{lcccc}
			
			\textbf{}  &\,  \textbf{1d, synch} \, & \, \textbf{1d, asynch} \, &\, \textbf{2d, synch} \,&\, \textbf{2d, asynch}  \, \\
			\hline\hline
			perfect measurements & $17.5\%$ & $3\%$ & $3.5\%$ & $1\%$\\ 
			\hline
			noisy measurements & $7.5\%$ & $2\%$ & $1.5\%$ &  $0.5\%$ \\ 
			\hline 			
		\end{tabular}
		\medskip 
		\caption{\label{tab:thresholds} Approximate thresholds under phenomenological noise with perfect measurements (top) and measurements with a noise rate equal to the qubit bit-flip rate (bottom); all estimates are rounded to the nearest $0.5\%$. The asynchronous decoders are simulated with Poissonian updates, and thresholds are higher when using the non-Markovian desynchronization scheme of sec.~\ref{sss:asynch_nonmark}. All quoted thresholds are obtained from numerics with $Z = \lceil \log_{3/2}(L) \rceil$. 
		}
	\end{table*}
	
	We also perform a similar analysis with perfect measurements, and for decoders with Poissonian asynchronous updates. We implement the latter by randomly choosing sites to update and performing defect-motion and message-passing automaton updates with probabilities $1/(1+v)$ and $v/(1+v)$, respectively, although it is very likely there exist other choices that yield higher thresholds. 
	Rough estimates for the values of $p_c$ obtained in these settings are tabulated in Table~\ref{tab:thresholds}. Poissonian updates  produce thresholds quite a bit lower than their synchronous counterparts, and (more significantly) have a reduced suppression in the sub-threshold error rate; an example for the $1d$ repetition code with noisy measurements is shown in fig.~\ref{fig:1dasynch}, where the values of $\G$ are approximately halved compared their synchronous counterparts. In the future it would be valuable to identify more noise-robust Markovian desynchronization schemes. 
	
	\begin{figure}
		\centering 
		\includegraphics[width=.44\textwidth]{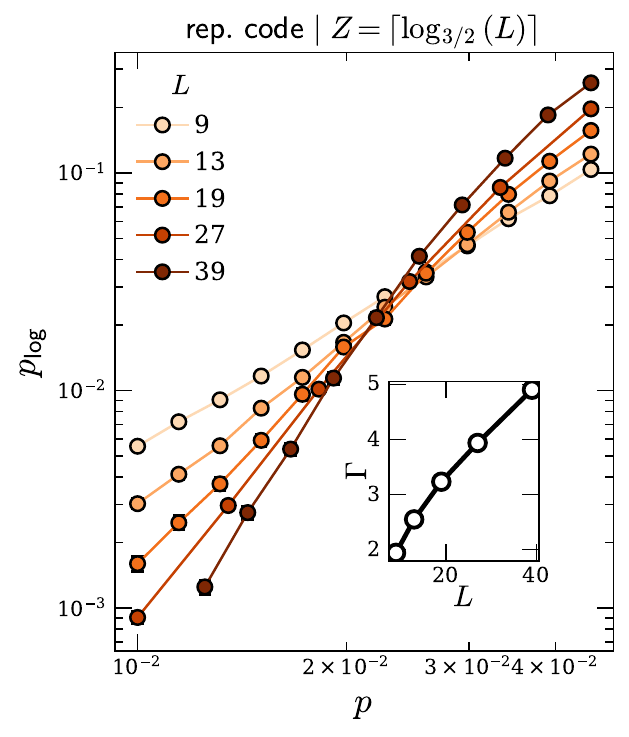}%
		\caption{\label{fig:1dasynch} Scaling of the logical error rate in the $1d$ repetition code with Poissonian asynchronous updates, under i.i.d bit-flip and measurement noise of strength $p$. As in fig.~\ref{fig:plog_vs_p} we take $Z = \lceil \log_{3/2}(L)\rceil$, and the inset plots a fit to $\plog = (p/p_c)^\G$ in the low-noise regime. }
	\end{figure}

	\ss{Phase transitions}
	
	We now briefly examine the phase transitions that occur as the memory is destroyed upon increasing $p$, which we numerically observe to be continuous. This is particularly interesting in the case of the $1d$ repetition code, as it yields a continuous phase transition with a local order parameter in one dimension.\footnote{The (logarithmically small) classical RG dimension should not be regarded as an extra dimension in this case, as the noise and order parameter are restricted to a completely $1d$ system (we also achieve stretched-exponential scaling of $\tmem$ with $L$ below threshold, which would not be achievable by e.g. placing Toom's rule on a cylinder of circumference $\log(L)$.} This is something which cannot happen in equilibrium, and the only other $1d$ cellular automaton with a continuous noise-driven transition not in the percolation universality class the author is aware of is Tsirelson's automaton \cite{balasubramanian2024local}, whose dynamics is strongly inhomogeneous (it is also possible to have phase transitions driven by changing the degree of synchronicity in the CA updates \cite{fates2007asynchronism,pajouheshgar2025exploring}; these transitions are different from the ones observed here).
	Future work will be needed to understand the nature of the phase transitions in our decoders, and here we content ourselves with a brief preliminary study of the transition in the $1d$ repetition code decoder. 
	
	\begin{figure*}
		\makebox[\textwidth][s]{%
			\hspace{\fill}%
			\includegraphics[width=.35\textwidth]{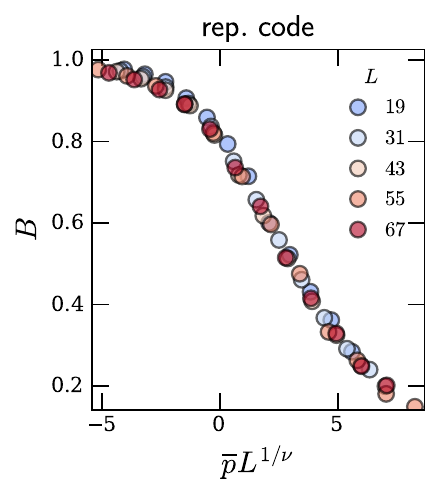}%
			\hspace{\fill}%
			\includegraphics[width=.36\textwidth]{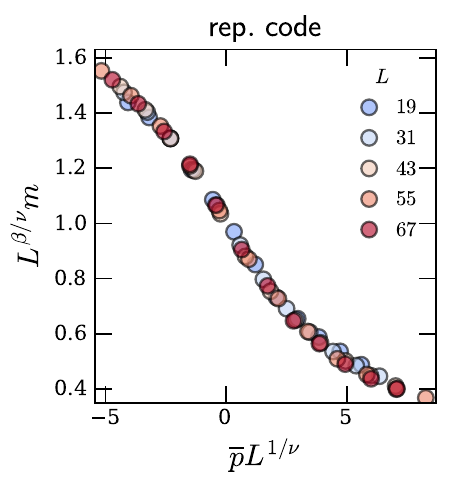}%
			\hspace{\fill}%
		}
		\caption{\label{fig:phase_transitions} Scaling collapses of the normalized Binder cumulant (left) and magnetization (right) for the 1d repetition code with synchronous updates; here $\bar{p} = p/p_c-1$ and the values of $p_c,\nu,\b$ are as in Tab.~\ref{tab:scaling_exponents}. }
	\end{figure*}
	
	As an order parameter for the phase transition, we use the expected value of the magnetization 
	\be\lan m \ran = \frac2L\lan   \bfs \cdot \bfs_{\sf corr} \ran - 1\ee 
	where $\bfs$ are the spins at $z=0$, $\bfs_{\sf corr}$ is the correction that the decoder would have applied at that time step, and $\lan \cdot \ran$ is taken in the non-equilibrium steady state. We also compute the fluctuations of the (absolute value of the) magnetization, defined as $\chi = L( \lan |m|^2\ran  - \lan |m|\ran^2 )$.
	Finally, to assist in identifying the critical value of $p$, we also compute the normalized Binder cumulant
	\be B = \frac32 - \frac{\lan m^4 \ran}{2\lan m^2\ran^2},\ee 
	which satisfies $\lim_{p \ra 0} B = 1, \lim_{p \ra 1/2} B = 0$ (the latter following from Wick's theorem), and is such that $B$ is independent of $L$ at $p_c$.

	\begin{table}[ht]
		\centering
		\renewcommand{\arraystretch}{2}
		\begin{tabular}{lcccc}
			
			\textbf{}  &  \textbf{$p_c$}\,\,\,&\,\,\,\textbf{$\nu$}\,\,\, &\,\,\,\textbf{$\b$}\,\,\, &\,\,\,\textbf{$\g$} \,\,\, \\
			\hline\hline
			synch & $7.5\%$ & $3/2$ & $1/4$ & $1$ \\ 
			\hline
			asynch & $2\%$ & $3/2$ & $1/8$ &  $ 5/4$ \\ 
			\hline 			
		\end{tabular}
		\medskip 
		\caption{\label{tab:scaling_exponents} Rough estimates for $p_c$ and scaling exponents in the 1d phase transition with (again taking measurement errors and bit-flip errors to occur with equal probability). The values are consistent with hyperscaling with $d=1$.
		}
	\end{table}
	
	We estimate the expectation values in $B,\lan m\ran $ in Monte Carlo, and perform a scaling collapse of $B$, $\lan m\ran L^{\b/\nu}$, and $\chi L^{-\g/\nu}$ against $\overline{p}L^{1/\nu}$, where we define $\overline{p} = (p_c - p)/p_c$. Rough estimates for the critical exponents $\b,\nu,\g$ obtained from this procedure are provided in  Table~\ref{tab:scaling_exponents} for both synchronous and asynchronous updates, and the scaling collapses for the magnetization and Binder cumulant are shown for synchronous updates in fig.~\ref{fig:phase_transitions}. The figure shows a decent scaling collapse for $\nu = 3/2, \b=1/4$, and we also find a respectable collapse of $\chi$ for $\g=1$, consistent with the hyperscaling relation $\b=(d\nu - \g)/2$ with $d=1$. When a similar analysis is performed for asynchronous updates, we find a fairly good collapse with the same value of $\nu=3/2$, but different values of $\b,\g$. 
	
	It is known in the literature that there exist critical points whose critical exponents depend on synchronicity. This has been studied for Toom's rule, where it was found that the exponents could depend on synchronicity, but with the ratios $\b/\nu,\g/\nu$ fixed to the values they take in the 2d Ising model \cite{ray2024protecting,takeuchi2006can}. A partial understanding for this result follows from the fact that the critical point of Toom's rule for asynchronous updates flows to the model-A universality class under RG \cite{ray2024protecting,squeezing}. As in the non-equilibrium phase transitions in the memories of ref.~\cite{squeezing}, more work will be needed to understand what happens in the present case, where the flow cannot be to any model with an effectively equilibrium description. 
	

	\ss{Initialization} \label{ss:ftinit}
	
	\begin{figure*}
		\makebox[\textwidth][s]{%
			\hspace{\fill}%
			\includegraphics[width=.36\textwidth]{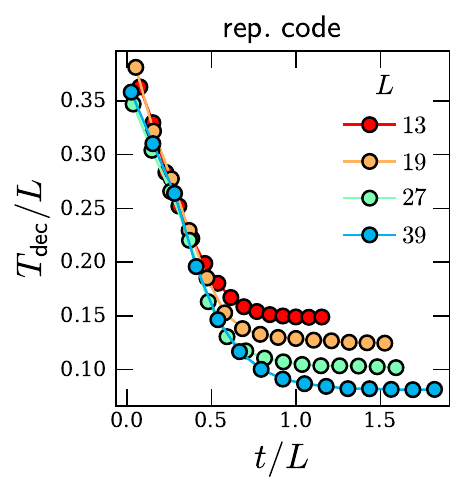}%
			\hspace{\fill}%
			\includegraphics[width=.35\textwidth]{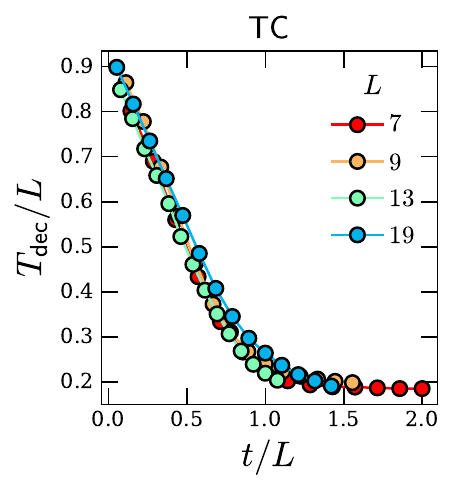}%
			\hspace{\fill}%
		}
		
		\caption{\label{fig:prep} Sub-threshold decoding times $\tdec$ (shorthand for $\tdec(t;\s_\tp(\rlog))$) for synchronous decoders following a quench from a random state, shown for the 1d repetition code with $p = 4.5\%$ (left) and the 2d toric code with $p = 0.75\%$ (right). In both plots we set $Z = \lceil \log_{3/2}(L) \rceil$.}
		
	\end{figure*}
	
	In addition to acting as a memory, a useful active decoder $\mcd$ should also be able to fault-tolerantly prepare logical states, reliably taking specific product states to targeted logical ones. We say that a decoder prepares a logical state $\rlog$ fault-tolerantly in time $t$ if there is a product state $\s_\tp(\rlog)$ such that the dynamics begun on $\s_\tp(\rlog)$ is $\ep$-close in trace distance to the dynamics begun on $\rlog$: 
	\be ||\mcd^t(\s_\tp(\rlog)) - \mcd^t(\rlog)||_1 < \ep \ee 
	for some small $O(1)$ constant $\ep$. We define the {\it initialization time} $\tprep$ of the decoder (for a fixed choice of $\s_\tp(\rlog)$) as (see also the discussion in
	\be \tprep = \max_{\rlog} \min\{t \, : \, ||\mcd^t(\s_\tp(\rlog)) - \mcd^t(\rlog)||_1 < \ep \}.\ee 
	
	For the present decoder, we analyze initialization following the discussion in \cite{Dennis_2002,balasubramanian2024local}. To prepare the $|\overline{00}\rangle\langle \overline{00}|$ logical state of the toric code, we choose $\s_\tp(|\overline{00}\rangle\langle \overline{00}|)$ to be the all-0 product state; the decoder is then initialized in a state with trivial $Z$ syndromes and $X$ syndromes which are equivalent to those one would obtain under maximally dephasing noise ($p=1/2$). Since the $X$ and $Z$ sectors of the decoder function independently, $\tinit$ is fixed by the time needed for the initial large density of $X$ syndromes to be corrected away. 
	
	Under the decoding dynamics, an initial state with $p=1/2$ will have a defect density which exhibits an initial exponential decay. Once the density becomes sufficiently low---and all the defects in the initial noisy state are pushed off to the back wall---we will be left with a small number of well-separated pairs, generically separated on scales of order $L$. The linear erosion property then suggests $\tprep = \ct(L)$. This will be true only if the remaining well-separated anyons on the back wall are able to establish direct causal contact. This only happens if the error rate $p_Z$ on the back wall satisfies $L < p_Z^{-1/D}$. If instead $L \gg p_Z^{-1/D}$, the incoming anyons will intervene and lead to the well-separated anyons moving diffusively, increasing $\tprep$ to $\ct(L^2)$ at the largest system sizes (of course, this is only relevant if we simultaneously have $\tmem > L^2$). 
	As a comparison, note that currently studied RG-based hardwired decoders have $\tprep = O(L^{\a>1})$ \cite{balasubramanian2024local}. 
	We cannot expect $\tprep = o({\rm poly}(L))$, since relaxation this fast would indicate that our decoder was phase-equivalent to a trivial channel, and we believe that $\tinit = \ct(L)$ is optimal for any local decoder. 
	
	Bounds on $\tinit$ can be proven by studying the erosion of system-spanning noise clusters, but the details are rather tedious. We will instead content ourselves by numerically evaluating a simple proxy for $\tinit$. We define the expected decoding time $\tdec$ as 
	\be \tdec(t;\r) = \min\{ T \, : \, \Tr[ \mcd_\off^T[\mcd^t(\r)] \Pi_\bfzero ] = 1 \},\ee 
	where $\Pi_\bfzero$ is the projector onto the defect-free subspace. The time $t_{\sf indep}$ after which $\tdec(t;\s_\tp(\rlog))$ becomes independent of $t$ defines the time after which system has effectively equilibrated as far as the readout of logical information (using the offline decoder) is concerned. We will thus take $t_{\sf indep}$ as a proxy for $\tinit$. 
	
	Fig.~\ref{fig:prep} shows $\tdec(t;\s_\tp(\rlog))$ at $p<p_c$ for both the $1d$ repetition code and $2d$ toric code. In both cases, we extract a value of $t_{\sf indep}$ that scales roughly as $L$, supporting the expectation of $\tinit = \ct(L)$ put forward above. 
	
	\section{Discussion and outlook} \label{sec:disc}
	
	In this work, we showed how to perform fault-tolerant active error correction in topological codes by using a local message-passing architecture to simulate a confining interaction between anyons. In what follows we briefly discuss some areas that would be interesting to address in future work. 
	
	{\it Fast real-time decoding:} 
	The parallelization and simplicity of our decoder means that very little time overhead is incurred when running error correction. This is of most relevance to superconducting platforms, where its simplicity opens up the possibility of the error-correcting hardware being implemented in-fridge. How close to on-chip it can get is determined by the amount of heat that the classical control architecture dissipates, which is ultimately set by total the number $n_{\sf bits}$ of classical bits used by the decoder (we need $n_{\sf bits} = \ct(L^2 \polylog(L))$ for our analytics to guarantee a maximal suppression of $\plog$, but our numerics indicate that $n_{\sf bits} = \ct(L^2 \log(L))$ is likely sufficient). In more detail, if we allow messages to propagate up to half the system size, we need $\lceil 6Z \log_2 (L/2) \rceil$ bits of memory per stabilizer location to store the messages, plus a few extra bits to store the syndrome value and potential feedback operations to be applied on incident edges, as well as to perform comparisons between different messages. Taking $L = 21$ and $Z = 5$, this gives $n_{\sf bits} \approx 10^5$ bits in total, although below threshold most of these bits would be inactive during the large majority of time steps (below threshold, the activity of a bit is exponentially suppressed in its $z$ coordinate, and only $\log_2 \log_2(L/2)$ bits are needed to represent messages created by typical errors).		
	This power demand should be modest enough to fit in the 4K stage of a dilution refrigerator; if the classical processing is implemented with single flux quantum logic or cryo-CMOS \cite{mukhanov2011energy,liu2023single,11037551,patra2017cryo} it may be possible to do better (e.g. on-chip) with a small enough value of $Z$, which in practice may be enough to achieve a sufficient amount of error suppression even if no strict threshold is present (an approach similar in spirit to the proposals of e.g.~\cite{park2025enhancing,holmes2020nisq+}).  
	Going forward, it will be worthwhile to make a detailed comparison between these resource requirements and those of decoders based on MWPM and Union-Find (see e.g.~\cite{liyanage2024fpga,chan2023actis,wu2023fusion,higgott2025sparse}), which, since they implement sliding-window schemes, require $n_{\sf bits} = \O(L^3)$.

	Of course, speed and efficiency are only one part of the story, and the accuracy of our approach is still appreciably below nonlocal strategies like MWPM and Union-Find. However, we have not attempted to optimize the failure rate in our construction, and future work will be needed to identify modifications of the message-passing architecture that improve the decoding accuracy; given that the naive construction presented here has a threshold under phenomenological noise that is only a factor of $\approx 2$ smaller than global MWPM (this comparison being made with synchronous message-passing dynamics), we are optimistic that such improvements can be made (one easy modification that improves performance is to add an extra rule that eliminates pairs of anyons separated by a single link, before they have a chance to send out messages).

	{\it Many-body physics:} From the perspective of quantum many-body physics, our construction gives a local dissipative process that is capable of maintaining quantum coherence for a thermodynamically long time in the presence of generic noise suffered by the quantum components of the system. It relies on coupling to a noiseless classical automaton to do this, and it would be interesting to understand the principles underlying the characterization of such classically-controlled quantum systems. In particular, for the message-passing decoders discussed in this work, the following questions are of interest: 
	
	\medskip 
	
	\begin{itemize} 
		\item How do we characterize the steady states of the dynamics in the error-correcting phase? 
		\item What are the universality classes of the phase transitions that occur when the noise strength is tuned, particularly in one dimension?
		\item Can we develop effective continuum descriptions for the anyon dynamics?
	\end{itemize}

	{\it Fault tolerance with no classical overhead:} Finally, the biggest outstanding open theory question left unaddressed by this work is the construction of a local dissipative decoder that uses {\it no} reliable classical bits. In particular, can one construct a decoder in $d\leq 3$ dimensions that is 1) local, 2) asynchronous, 3) homogeneous in spacetime, and 4) robust also against noise in any classical bits used to assist quantum error correction? 
	The construction in this work satisfies criteria 1-3 but not 4. 
	Of the other decoders in the literature, the $3d$ construction of \cite{balasubramanian2024local} comes closest: it only lacks homogeneity in space and robustness against desynchronization (and has the benefit of storing $\ct(L)$ logical qubits), and these problems may be easier to overcome. Finding out how to solve these problems---and ideally, doing so in a way that is practically useful---is an exciting challenge for future work.

	\section{Acknowledgments}
	
	I thank Ehud Altman, Tarun Grover, Jeongwan Haah, Yuri Lensky, Yaodong Li, and Charles Stahl for helpful discussions and suggestions, Shankar Balasubramanian, Aditya Bhardwaj, Margarita Davydova, and  Nathaniel Selub for discussions and collaboration on related work, and Dominic Williamson for alerting me to David Poulin's unpublished work on simulated confinement. I am supported by a Miller research fellowship. 
	
	\appendix 
	
	\begin{widetext}

		\section{No-go results}\label{app:nogo}
		
		\ss{Decoders with message-passing dynamics}
		
		In this subsection, we show the following result: 
		\ms 
		\begin{theorem}[no threshold for the message-passing decoder with $Z=0$]\label{thm:no_thresh_msg_passing} 
			There exists a $p$-bounded error model for which the local message-passing decoder with $Z=0$ has $\tmem$ bounded as in \eqref{z0tmembound}. This holds in any dimension, and in the absence of measurement errors. 
		\end{theorem}

		As described in the main text, it suffices to show that the result holds in $1d$; the result for $d>1$ then follows immediately by considering a noise model that acts only on a $1d$ strip of the lattice and using the isotropy properties of the message passing architecture. The proof proceeds by designing a $p$-bounded error model that always tries to ``expand'' well-separated anyon pairs. Since we are interested in bounding $\lim_{L\ra\infty}\tmem$, we will directly work in the thermodynamic limit. Finally, for convenience we will work with synchronous updates throughout; the extension to asynchronous updates involves more bookkeeping details but no new ideas. 
		\ms\begin{definition}[single-pair error model]\label{def:singlepair_noise_model}
			The {\it single-pair error model} $\mce_{\sf s.p.}$ is a particular $p$-bounded error model defined as follows:
			\begin{enumerate}
				\item Suppose that at time $t$ there is an anyon pair at coordinates $r_{l,t},r_{r,t}$ with $r_{l,t} < r_{r,t}$ and separation $r_t = r_{r,t} - r_{l,t} >2$. Then the noise $\sfN_t$ at time $t$ creates only localized anyon pairs which are separated from one another and from the anyons in the pair by at least two lattice sites, whose spatial arrangement is reflection-symmetric about the midpoint of the anyon pair, which are created only at coordinates $r>r_{r,t}$ and $r<r_{l,t}$,
				and which are otherwise created independently with probability $p$. Define the following set of pairs of lattice sites: 
				\be \label{noisesites} {\sf NoiseSites}_t = 
				\{(r_{l,t}-3n+1,r_{r,t}+3n) \, : \, n \in \zz^{>0} \},\ee
				then with probability $p$ independently for each tuple $(r_1,r_2) \in {\sf NoiseSites}_t$, $r_1,r_2\in \sfN_t$. If a site $r$ is not contained in any tuple in ${\sf NoiseSites}_t$, then $r \not\in \sfN_t$. 
				
				\medskip 
				
				\item If at time $t$ there is no anyon pair in the system with separation greater than 2, then $\sfN_t$ creates one pair of separation $r_0$ centered at the origin with probability $p^{r_0}$, with $r_0$ an odd integer greater than 2, and does not create any other pairs. 
			\end{enumerate}
		\end{definition}\ms
		
		The proof is conceptually very simple---it is ultimately just a statement about a biased random walk on $\zz^{\geq 0}$---but there are unfortunately a few details to be kept track of. 
		
		We will first show that the noise rapidly creates a ballistically-expanding pair when the message velocity is $v=\infty$. We will then describe how to extend this result to the local version where $v$ is an $O(1)$ constant. To show the $v=\infty$ case, we use the following proposition: 
		\begin{proposition}\label{prop:biased_pair_sep}
			Let $r_t$ be the separation between the anyons in the largest anyon pair at time $t$. For the message-passing decoder with $v=\infty$ subjected to $\mce_{\sf s.p.}$, the variable $r_t$ satisfies 
			\be \sfP(r_{t+1} = r_t\pm 2) = \frac{1\pm b_{r_t}}2,\qq \sfP(r_{t+1} = r_t) = 0,\ee 
			where the bias of the noise is 
			\be b_r = 1 - 2 (1-p)^{\lfloor \frac{r+2}4 \rfloor},\ee 
			and where $\sfP(\cdot)$ is calculated by averaging over $\sfN_t$. 
		\end{proposition}
		\begin{proof}
			The result follows from unpacking the definition of $\mce_{\sf s.p.}$ and the fact that when $v=\infty$, each anyon moves towards its instantaneous nearest neighbor (it will not be necessary to decide what to do in the case of ties). From the choice of ${\sf NoiseSites}_t$ in \eqref{noisesites}, as well as the fact that the created anyon pairs are always separated by an odd number of lattice sites (in particular, that $r_t \in 2\zz +1$; c.f. the second point of the above definition), it follows that $r_{t+1} = r_t \pm 2$ for all $t$, and that every anyon in the system always has a unique nearest neighbor. Thus $\sfP(r_{t+1}=r_t) = 0$. 
			
			To calculate $\sfP(r_{t+1} = r_t - 2)$, we observe that $r_{t+1} = r_t - 2$ only if the noise does not create pairs at any of the $\lfloor \frac{r+2}4\rfloor$ possible elements of ${\sf NoiseSites}_t$ with second coordinate between $r_{r,t}$ and $r_{r,t}+r_t$. Therefore 
			\bea \sfP(r_{t+1} = r_t-2) = (1-p)^{\lfloor \frac{r_t+2}4 \rfloor},\eea 
			which gives the stated expression for $b_t$. 
		\end{proof}
		
		We can now complete the proof of theorem~\ref{thm:no_thresh_msg_passing}. Since the details are slightly tedious and the analysis is very elementary, we will be slightly schematic at places. As previously, we will let undefined roman letters denote unimportant $p$-independent constants, which we will let change values in between equations. 
		
		\begin{proof} 
			Under $\mce_{\sf s.p.}$, the large pair separation $r_t$ undergoes a random walk on $2\zz^{\geq 0}+1$. From Prop.~\ref{prop:biased_pair_sep}, this walk is biased towards large $r_t$ when 
			\be \label{rstar} r_t \geq r_* = 4 \frac{\ln2 }{\ln[1/(1-p)]} - 2,\ee 
			and is biased towards small $r_t$ when $r< r_*$. To simplify the notation a bit, consider equivalently a walk on $\zz^{\geq0}$ where $\sfP(r_{t+1} = r_t \pm 1) = (1\pm b_r)/2$, with $b_r =1 - a\g^r$ and $\g=(1-p)^{1/4}$. The expected time to hit $\wt r_* = 10 \ln(1/a)/\ln(1/\g)$ (where the walk is comfortably biased towards $\infty$) starting from the origin has expectation value (see e.g. \cite{levin2017markov})
			\be \label{hitting} \oEE[\tau_{\wt r*}] \leq a e^{b/\ln(1/(1-p))} + c.\ee 
			
			By itself, this noise model will not create an error on a system of size $L$ until a time going like $\oEE[\tau_{\wt r_*}] + \ct(L)$, since it takes $\ct(L)$ time for the pair to expand to a separation of order $L$, and there is only ever one large pair present in the system. To bring down the expected time to create an error to something proportional to $\oEE[\tau_{\wt r_*}] $, we simply divide the system into blocks of width $a\wt r_*$, and define a noise model $\mce$ which implements $\mce_{\sf s.p.}$ on each block. After a time $b\oEE[\tau_{\wt r_*}]$, a large pair will be present on each block with probability at least $9/10$ (for an appropriate choice of $b$), and decoding with the message passing decoder (or a global majority vote, in the case of the $1d$ repetition code) will produce a logical failure with probability 1 in the $L\ra\infty$ limit; combining this with \eqref{hitting} then gives what we wanted to show.\footnote{A particularly careful and / or pedantic reader will notice that we have swept some annoying details about what happens near the boundary between blocks under the rug: the noise added by the $\mce_{\sf s.p.}$s on each block must be terminated by the block boundaries, and this allows the large anyon pairs from adjacent blocks to fuse with one another. These effects just decrease $\tmem$, and we will not bother to spell them out explicitly.}
			
			\begin{figure}
				\includegraphics[width=.5\tw]{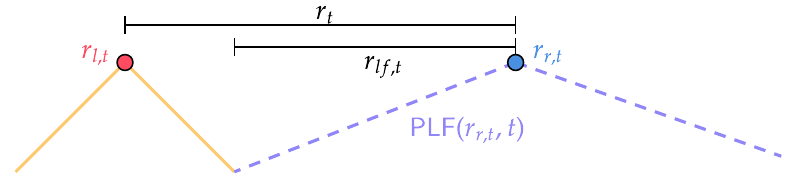} \caption{\label{fig:geo_argument} The considerations relevant to showing \eqref{rlct}; time runs vertically upward. The dashed lines denote $\plc(r_{r,t},t)$ and are drawn with slope $1/v$.}
			\end{figure}
			
			Now we need to generalize this argument to a constant message velocity.  
			Consider again the random walk that a large anyon pair undergoes under $\mce_{\sf s.p.}$.
			The probability for the pair separation to decrease at a given time is now
			\be \label{vprob}\sfP(r_{t+1} = r_t -2) = (1-p)^{\lfloor \frac{r_{lf,t} + 2}4 \rfloor },\ee 
			where		
			\be r_{lf,t} = |r_{r,t} - \{r_{l,t'<t}\} \cap \plc(r_{r,t},t)|,\ee 
			where $\plc(r,t)$ is the past light front (with respect to the velocity-$v$ messages) of the spacetime point $(r,t)$. 
			By the reflection symmetry of the noise we additionally have $r_{lf,t} =  |r_{l,t} - \{r_{r,t}\} \cap \plc(r_{l,t'<t},t)|$; thus having $r_t \in 2\zz+1$ for all $t$ holds independent of $v$. To get an upper bound on $\tmem$ it suffices to use the upper bound 
			\be\label{rlct} r_{lf,t} \geq \left\lfloor \frac{vr_t}{1+v}\right\rfloor,\ee 
			which follows from a simple geometric argument shown in fig.~\ref{fig:geo_argument}. Replacing $r_{lf,t}$ in \eqref{vprob} by the RHS of \eqref{rlct} gives a random walk that moves to the right with strictly less probability than the actual random walk of interest, and hence an upper bound on $\oEE[\tau_{\wt r_*}]$ for this random walk can be used to upper bound $\tmem$. The new random walk however clearly still has a hitting time like that in \eqref{hitting} (just with different $p$-independent constants), and so taking $v < \infty$ does not qualitatively change the result. 
		\end{proof}

		\ss{Decoders with power-law interactions} \label{app:power_law_nogo}
		
		We now show the analogue of theorem~\ref{thm:no_thresh_msg_passing} for decoders with attractive power-law interactions between anyons. 
		
		\ms 
		
		\begin{theorem}[no threshold for power-law decoders]\label{thm:no_thresh_powerlaws} 
			There exists a local $p$-bounded error model for which the power-law decoders have $\tmem$ bounded as in \eqref{z0tmembound}. This holds in any dimension, for any value of $\a$, and in the absence of measurement errors. 
		\end{theorem}
		
		\begin{proof}
			The proof preceeds in the same way as the proof of theorem~\ref{thm:no_thresh_msg_passing} and uses the same noise model, and so we will be rather succinct. 
			
			Consider again the effects of $\mce_{\sf s.p.}$ in the $v=\infty$ limit. We just need to establish the analogue of prop.~\ref{prop:biased_pair_sep}. To calculate $\sfP(r_{t+1} = r_t-2)$, focus on the force at the location of the right anyon in the pair $F_{r_{r,t}}$. Then $\sfP(r_{t+1} = r_t -2) = \sfP(F_{r_{r,t}} > 0)$. Due to the reflection symmetry of $\mce_{\sf s.p.}$, $F_{r_{r,t}}$ recieves a positive contribution from every small pair created by $\sfN_t$. We can thus lower bound $\sfP(F_{r_{r,t}} < 0)$ by considering only the pair created by $\sfN_t$ which is nearest to $r_{r,t}$, together with the anyon at $r_{l,t}$. If the nearest pair is at location $r_{r,t} + \d$ with $\d$ a random variable determined by $\sfN_t$, then the force created by it, its reflection-related partner, and the anyon at $r_{l,t}$ gives 
			\bea \sfP(F_{r_{r,t}} < 0) & \leq \sfP\( \frac1{\d^\a} + \frac1{(\d+1)^\a} - \frac1{r_t^\a} - \frac1{(r_t + \d)^\a} - \frac1{(r_t+\d+1)^\a} < 0\)\\ 
			& \leq \sfP\( \frac2{(\d +1)^\a}- \frac 3{r_t^\a} < 0 \) \\ 
			& = \sfP\( \d > r_t(2/3)^{1/\a} - 1\)\\
			& \leq (1-p)^{\lfloor 2^{1/\a} 3^{-1/\a-1} r_t \rfloor }\eea 
			where the last inequality is the probability for the $\lfloor (\d +1)/3 \rfloor$ potential sites at which $\sfN_t$ could create a pair between $r_{r,t}$ and $r_{r,t}+\d$ to all be noise-free. 
			Therefore the bias of the noise satisfies 
			\be b_r \geq 1 - 2 (1-p)^{\lfloor 2^{1/\a} 3^{-1/\a-1} r \rfloor }.\ee
			Since $b_r$ is again exponentially close to 1 at large $r$, the remainder of the proof can be done by exactly following that of theorem~\ref{thm:no_thresh_msg_passing}.
		\end{proof}

		\ss{Power-law interactions with charged anyons}\label{app:charged}
		
		We now show that endowing anyons with $U(1)$ charge does not help construct a decoder that operates by way of power-law interactions, even when the associated conservation law is a strong symmetry obeyed by the noise. In this model, the force on the $i$th anyon is 
		\be \bfF_{i} = q_i \sum_{j\neq i} q_j\frac{\bfr_j - \bfr_i}{|\bfr_i - \bfr_j|^{\a +1}},\ee 
		where the charges $\{q_i\}$ will be assumed for simplicity to be valued in $\pm1$. 
		
		Consider as before a well-separated anyon pair at locations $\bfr_1,\bfr_2$, and let the remaining anyon pairs by at locations $\bfr_{i,1}, \bfr_{i,2}$. Since paired anyons always have opposite charge (by imposition of the strong $U(1)$ symmetry), we may re-organize the sum in the expression for the force $\bfF_1$ on anyon $1$ as 
		\be \bfF_{1} = \frac{\bfr_2 - \bfr_1}{|\bfr_2-\bfr_1|^{\a +1}} + \bfxi_{\bfr_1},\ee 
		where now 
		\be \bfxi_{\bfr_1} = \sum_i \frac{\bfd_i}{|\bfr_1 - (\bfr_{i,2} + \bfr_{i,1})/2|^{\a +1}},\ee 
		where the dipole moment $\bfd_i$ has magnitude of order $|\bfr_{i,1} - \bfr_{i,2}|$. Assuming that the anyon pair separation can be treated as a random variable with $O(1)$ variance, the variance $\s^2_F$ of the noise is then suppressed relative to the uncharged case, with 
		\be \s^2_F \sim \int d^dr \, \frac1{r^{2(\a +1)}},\ee 
		which now has a divergence at large $r$ only when $\a \leq d/2-1$. Unlike the uncharged case, we now have room to choose $\a < 1$ without $\s_F^2$ diverging, and so the assumption that the pair separation has finite variance has the possibility of being self-consistent. 
		
		These models nevertheless still do not possess robust thresholds. For i.i.d noise, the charged nature of the interactions does not help once the classical communication is rendered non-instantaneous: any amount of retardation results in anyons recieving a transient non-screened signal from a pair created by the noise, which is enough to re-introduce a divergence to $\s_F^2$ when $\a \leq 1$ (numerics interestingly appear to indicate that charged models actually perform {\it worse} than uncharged ones). 
		
		We may also rigorously show that charged models lack a threshold under a particular $p$-bounded error model, again proceeding along the same lines as in the proof of theorem~\ref{thm:no_thresh_powerlaws}. We take the noise model to create at most one large pair with opposite charges $q_l,q_r$. We furthermore let the pairs created by $\sfN_t$ at  to be such that if $q_i$ is the charge of the $i$th anyon in the system with the anyons ordered according to their $r$ coordinate, then $q_i = -q_{i+1}$. Following the logic in the proof of theorem~\ref{thm:no_thresh_powerlaws} and using the same notation, we have 
		\bea \sfP(r_{t+1} = r_t-2) & \leq \sfP\(\frac1{\d^\a} - \frac1{(\d+1)^\a} - \frac1{r_t^\a} + \frac1{(r_t+\d)^\a} - \frac1{(r_t + \d + 1)^\a} < 0\) \\
		& \leq \sfP\(\frac1{\d^\a} - \frac1{(\d+1)^\a} - \frac1{r_t^\a} <0\) \\
		& \leq \sfP\(\frac{\min(\a,1/\a)}{\d^{\a+2}}- \frac1{r_t^\a} \) \\ 
		& = \sfP\( \d < C_\a r_t^{1/(1+2/\a)}\) \\ 
		& \leq (1-p)^{\lfloor \frac{C_\a r_t^{1/(1+2/\a)} + 1}3 \rfloor},
		\eea 
		where we defined $C_\a = \min(\a,1/\a)^{1/(2+\a)}$ and used the inequality 
		\be \frac1{\d^\a} - \frac1{(\d+1)^\a} > \frac{C_\a^{2+\a}}{\d^{\a+2}},\qq \d \in \zz^{\geq 0} + 2,\ee 
		which follows from some unilluminating calculus. 
		The bias of the resulting random walk now converges to 1 only as a stretched exponential function of $r_t$, but this is still (more than) enough to unbind the random walk undergone by $r_t$, and leads to a similar type of scaling \eqref{tmemscaling} as derived previously. 
		
		\section{Proof details} \label{app:details} 
		
		\ss{Proof of lemma~\ref{lemma:cluster_erosion}}
		We now prove lemma \ref{lemma:cluster_erosion} from the main text. Since the proof strategy is essentially identical to the analogous result in \cite{lake2025fast}, we will be somewhat laconic. 
		
		\begin{proof} We proceed by induction on the cluster level $k$, showing that the defects in each cluster expand at most by an amount proportional to the linear cluster size before messages are exchanged between defects it contains, after which point it contracts until all the defects it contains mutually annihilate. If $b/w$ is chosen to be large enough, the contraction of each cluster is completed before it has a chance to link up with any other cluster, provided it does not survive to reach the back wall.
			
			Consider the base case of $k=0$, and let $\mcc^{(0)}$ be a particular $0$-cluster. All the defects in $\mcc^{(0)}$ will be guaranteed to be fully ``loaded'' into the control architecture (viz. will all be at $z>0$) over a time window of length $\leq w$, after which they will be contained within an $\infty$-norm ball of size no more than $3w$, regardless of the message background\footnote{Here and below we will abuse notation by describing the defects created by the noise in $\mcc^{(0)}$ as being ``in $\mcc^{(0)}$'' even after they have been loaded into the control architecture at $z>0$.} (the defects in $\mcc^{(0)}$ can expand by up to $w$ in all spatial directions, and move with maximum velocity $2\pm1$ along $\pm\uvz$, due to the modified rules discussed above). After a further time of $3w/(v-1)$, any defect in $\mcc^{(0)}$ at a spacetime location $u$ is guaranteed to have another defect from $\mcc^{(0)}$ in its past lightfront $\plf(u)$, defined as the set of all spacetime points $u'$ such that a message emitted from $u'$ can reach $u$. Provided $b/w>c$ is sufficiently large (e.g. $c \geq 10+12/(v-1)$ suffices), then after at most $w(1+3/(v-1))$ time steps, 
			\begin{itemize} 
				\item Each defect in $\mcc^{(0)}$ cannot have fused with a defect in any other spacetime noise cluster: this follows from $\mcc^{(0)}$ being spacetime separated from other noise clusters by a sufficiently large amount, and the fact that defects must always move to larger $z$ coordinates at each time step.
				\item The closest defect in $\plf(u)$ of any defect in $\mcc^{(0)}$ at the spacetime point $u$ is guaranteed to also be a (distinct) defect in $\mcc^{(0)}$. 
			\end{itemize} 
			By the linear erosion property proved in \cite{pajouheshgar2025exploring}---which shows that if $v>2$, a well-isolated cluster is guaranteed to be eliminated in a time linear in its size---the defects in $\mcc^{(0)}$ will thus annihilate against other defects in $\mcc^{(0)}$ after at most another $w(4+3/(v-1))$ time steps. In particular, they will be corrected before they reach the back wall provided $Z > c$ for some (other) constant $c$. 
			
			Assume now that the statement about cluster annihilation is true at all levels $k'<k$, and consider a $k$-cluster $\mcc^{(k)}$. By the inductive assumption, its correction can be analyzed while assuming that the spacetime noise is $(k-1)$-sparse, provided that it does not reach the back wall. For the same reasons as in the base case, the defects in $\mcc^{(k)}$ are guaranteed to be annihilated against each other by time  $cn^k$, provided $cn^k < Z$ and $b/w$ is again sufficiently large. The induction hypothesis thus holds provided $k < k_Z = \log_n(Z/c)$ (it fails when $k\geq k_Z$, since in this case $\mcc^{(k)}$ can reach the back wall before being corrected). 
		\end{proof} 
		
		\ss{Proof of lemma~\ref{lemma:kmax}}
		
		We now prove lemma~\ref{lemma:kmax} from the main text. Our argument will yield a fairly poor bound on $\g$, but this will not matter for our present purposes. In the proof, undefined roman letters (viz. those other than $b,n,w,d,v,p$) will be stand-ins for $O(1)$ positive constants. To simplify the proof we will also work in the $v\ra\infty$ limit; the generalization to $v = O(1)$ (which only slightly slows down error correction; c.f.~\cite{lake2025fast}) is straightforward but tedious. 
		
		\begin{proof} 
			Consider building a clump $\scc$ from $m$ spacetime clusters $\{\mcc_{\a}^{(k_\a)}\}$, and set $\kmax = \max_\a k_\a$. To upper bound the linear size $r(\scc)$ of $\scc$, we may without loss of generality consider the case where $k_\a = \kmax\,\, \forall \, \, \a$ (we will accordingly omit the superscripts on clusters in what follows). 
			
			As long as $b/w$ is sufficiently large, if $\mcc_\a,\mcc_\b$ are distinct $\kmax$-clusters, no defect from $\mcc_\a$ can fuse with one in $\mcc_\b$ before the defects from both clusters reach the back wall. Clumps are thus formed by well-separated events occurring when the defects from distinct noise clusters are deposited on the back wall.  
			
			To each clump $\scc$ we may define a binary\footnote{Strictly speaking the tree need not be binary if more than two clumps join together at a single time step, but resolving $(n>3)$-valent nodes intro trivalent ones can be done without changing the maximum size of clumps that can be generated.} tree $\sfT(\scc)$ capturing the way in which $\scc$ was formed by back-wall cluster deposition: the leaves in $\sfT(\scc)$ are associated to noise clusters, the nodes are clumps, and the parents of a node are the smaller clumps from which it was formed. 
			
			Consider now two clumps $\scc_1,\scc_2$ in the vertices of $\sfT(\scc)$ which are the children of a common parent cump $\scc_{1\odot2}$. Suppose that at the time $t_0$ when a defect in $\scc_{1/2}$ first had as its nearest neighbor a defect in $\scc_{2/1}$, these clumps had linear sizes $r(\scc_1)$ and $r(\scc_2)$, respectively. Under $v=\infty$ message-passing dynamics, the maximal time $t_{\sf max}$ for which these clumps can expand before merging is 
			\be t_{\sf max} \leq c \min(r(\scc_1),r(\scc_2)),\ee 
			and therefore the size $r(\scc_{1\odot2})$ of $\scc_{1\odot2}$ when it forms is at most 
			\be \label{cluster_combination} r(\scc_{1\odot2}) \leq  \max(r(\scc_1),r(\scc_2)) + c \min(r(\scc_1),r(\scc_2)),\ee 
			where $c>1$. 
			
			
			
			We now upper bound the largest size of $\scc$ over all choices of trees $\sfT(\scc)$ with $m$ leaves, each of which is a $\kmax$ cluster. To simplify the analysis, assume instead that we are instead given a tree with $\lceil \log_2(m) \rceil$ leaves. The largest size of the root is obtained by maximally balancing the tree, since \eqref{cluster_combination} is maximized at fixed $r(\scc_1)+r(\scc_2)$ when $r(\scc_1) = r(\scc_2)$. 				
			Then the clump size of a node at depth $\sfd$ into the tree is at most  $wf^\sfd n^\kmax$ (where $f=1+c$), and hence 
			\be r(\scc) \leq wn^{\kmax} f^{\lceil \log_2(m) \rceil}, \ee 
			giving  \eqref{kmax} with $\g = \log_2(f)$. 
		\end{proof} 
		
		\ss{Proof of theorem~\ref{thm:typ_gerry}}
		
		We now give the proof of theorem~\ref{thm:typ_gerry} from the main text. For simplicity, we will present the argument for offline decoding only, showing that the failure rate of the offline message-passing decoder scales as \eqref{ploggerry} at small enough $p$. The generalization to active decoding is nearly immediate, since the same noise configurations we consider below will lead to a logical failure in the active setting. We will also simplify the proof by taking $v=\infty$; the same arguments go through when $v$ is finite but more $v$-dependent constants need to be tracked (the needed inequalities are in \cite{lake2025fast}). 
		
		\begin{proof} 
			Showing that $d_\mcd = o(d_\mcc)$ was done in \cite{lake2025fast}, where certain ``Cantor strings'' of weight $o(L)$ were created that were shown to produce a logical failure when fed into the offline message-passing decoder (see also \cite{wootton2015simple,hutter2015improved} for a similar construction). These error strings are created by beginning with a system-spanning error string, breaking the string into $q$ equal-length segments, deleting the last segment from the string, and then iterating on the remaining $q-1$ segments. When this iteration is continued until all remaining strings individually have $O(1)$ weight, the resulting error pattern has weight 
			\be L \cdot \( \frac{q-1}q\)^{\log_q(L)} = L^{\g} = o(L), \qq \g = \log_q(q-1)\ee 
			with the explicit construction in \cite{lake2025fast} showing that a logical failure is produced with certainty (regardless of $v$) if $q \geq 6$. 
			
			This result by itself does not show \eqref{ploggerry}: there are only $O(L)$ error patterns whose noise contains a single Cantor string, and {\it typical} error patterns---which are the ones germane to the present discussion since the $L\ra\infty$ limit is being taken first---must therefore be analyzed in order to show the desired result. Nevertheless, we may use clustering to show that, as long as $p$ is sufficiently small, typical error patterns for i.i.d noise which cause a logical failure are indeed likely to contain a pattern equivalent in some sense to a Cantor string. 
			
			Let us first define a ``coarsened'' Cantor string $\mcs$ obtained by only iterating until all remaining strings have length at least $(\log L)^\l$, with the value of $\l$ to be chosen later. After $k$ iterations, the length of the remaining string segments is $(q-1)Lq^{-k}$, and so we accordingly stop the iteration at level\footnote{Omitting floor functions here and below to declutter the notation.}
			\be k_\star = \log_q(L\cdot (q-1)) - \log_q(\l \log L),\ee 
			which is quite close to the maximum number of iterations $\log_q(L)$. 
			The weight $w(\mcs)$ of $\mcs$ is thus only slightly increased, and in particular remains $o(L)$:
			\be w(\mcs) = L \cdot \(\frac{q-1}q\)^{k_\star} = L^\g (\log L)^{\l(1-\g)}.\ee 
			
			We will now show that $\mcs$ has a relatively large probability of being present in i.i.d noise. We do this by conditioning on $\sfN \setminus \mcs$ being composed only of clusters of size $<(\log L)^\d$ for some $\d < \l$; that is, we condition on 
			\be \label{noisecond} (\sfN \setminus \mcs) \cap \sfN_{k_\d} = \emp, \quad k_\d = \log_n[(\log L)^\d]. \ee 
			When this is done, the linear erosion property proved in \cite{lake2025fast} (defects in a $k$-cluster $\mcc^{(k)}$ annihilate only on other defects in $\mcc^{(k)}$) means that endpoints of the strings in $\mcs$  move by a distance $o((\log L)^\l)$ before all errors except these strings are cleaned up. Once the defects in $\sfN\setminus \mcs$ have been eliminated, we claim that the defects in $\mcs$\footnote{Again, we will abuse notation by referring to a defect as being ``in $\mcs$'' if it was either in $\mcs$ initially, or fused with a defect in $\mcs$.} then attract each other and fuse in exactly the same pattern that they would fuse when $\sfN = \mcs$ (and consequently produce a logical error).This can be shown on induction in the level of iteration used to produce $\mcs$. At the deepest level---where defects are separated by amounts of order $(\log L)^\l$---this follows from linear erosion and the fact that the positions of the defects in $\mcs$ differ from their initial positions by only $o((\log L)^\l)$ (this is where taking $v \ra\infty$ simplifies things slightly). After the first round of $\mcs$ defect fusions, the remaining defects are now separated by an amount of order $q (\log L)^\l + o((\log L)^\l)$, and the same argument goes through: at each level, the correction to the defect positions is parametrically smaller in $L$ than the separation they have in the case where $\sfN = \mcs$. 
			
			Clustering, a union bound, and the i.i.d nature of the noise means that the probability of $\sfN$ being appropriately sparse away from $\mcs$ is 
			\be \sfP( (\sfN \setminus \mcs) \cap \sfN_{k_\d}\neq  \emp \, | \, \mcs \subset \sfN) \leq L^d (p/p_c)^{(\log L)^{\d\log_n(2)}},\ee 
			which vanishes as $L\ra\infty$ provided $\l > \d \geq  1/\a + \ep$ for some positive constant $\ep$. 
			We may therefore take e.g. $\l = \d + \ep$ to show that
			\be \sfP(\exists \, \text{Cantor string $\mcs$} \subset \sfN \, \wedge \, (\sfN \setminus \mcs) \cap \sfN_{k_\d} = \emp ) = p^{o(w(\mcs)^{1+\ep})} = p^{o(L^{\g + \ep })},\ee  
			which implies the desired statement about $\tmem$, since any noise realization $\sfN$ that contains $\mcs$ and is appropriately sparse when $\mcs$ is excised from $\sfN$ leads to a logical error.				
		\end{proof} 
		
	\end{widetext} 
	\bibliographystyle{unsrt}
	\bibliography{online_decoders}
	
\end{document}